\documentclass[hidelinks,onefignum,onetabnum]{siamart251104}

\usepackage{microtype}
\usepackage{booktabs}
\usepackage{enumitem}
\setlist[itemize]{topsep=0pt, partopsep=0pt}

\usepackage{lipsum}
\usepackage{amsfonts}
\usepackage{graphicx}
\usepackage{epstopdf}
\usepackage{algorithmic}
\usepackage{subcaption}
\usepackage{bm}
\ifpdf
  \DeclareGraphicsExtensions{.eps,.pdf,.png,.jpg}
\else
  \DeclareGraphicsExtensions{.eps}
\fi

\newsiamremark{remark}{Remark}
\newsiamremark{hypothesis}{Hypothesis}
\crefname{hypothesis}{Hypothesis}{Hypotheses}
\newsiamthm{claim}{Claim}
\newsiamremark{fact}{Fact}
\crefname{fact}{Fact}{Facts}

\headers{grouped pooled posterior via ensemble Kalman inversion}{H. Yang, X. Dong, JL. Wu.}

\title{Bayesian experimental design: grouped geometric pooled posterior via ensemble Kalman methods\thanks{Submitted to the editors DATE.
\funding{This work is supported by the University of Wisconsin-Madison, Office of the Vice Chancellor for Research and Graduate Education with funding from the Wisconsin Alumni Research Foundation.}}}

\author{
Huchen Yang\thanks{Department of Mechanical Engineering, University of Wisconsin--Madison, Madison, WI
  (\email{huchen.yang@wisc.edu}, \email{xdong94@wisc.edu}, \email{jinlong.wu@wisc.edu}).}
\and
Xinghao Dong\footnotemark[2]
\and
Jin-Long Wu \footnotemark[2]
}

\usepackage{amsopn}

\ifpdf
\hypersetup{
  pdftitle={Bayesian experimental design: grouped geometric pooled posterior via ensemble Kalman methods},
  pdfauthor={H. Yang, X. Dong, JL. Wu.}
}
\fi

\DeclareMathOperator*{\argmax}{arg\,max}

\begin{document}

\maketitle

\begin{abstract}
Bayesian experimental design (BED) for complex physical systems is often limited by the nested inference required to estimate the expected information gain (EIG) or its gradients. Each outer sample induces a different posterior, creating a large and heterogeneous set of inference targets. Existing methods have to sacrifice either accuracy or efficiency: they either perform per-outer-sample posterior inference, which yields higher fidelity but at prohibitive computational cost, or amortize the inner inference across all outer samples for computational reuse, at the risk of degraded accuracy under posterior heterogeneity. To improve accuracy and maintain cost at the amortized level, we propose a grouped geometric pooled posterior framework that partitions outer samples into groups and constructs a pooled proposal for each group. While such grouping strategy would normally require generating separate proposal samples for different groups, our tailored ensemble Kalman inversion (EKI) formulation generates these samples without extra forward-model evaluation cost. We also introduce a conservative diagnostic to assess importance-sampling quality to guide grouping. This grouping strategy improves within-group proposal–target alignment, yielding more accurate and stable estimators while keeping the cost comparable to amortized approaches. We evaluate the performance of our method on both Gaussian-linear and high-dimensional network-based model discrepancy calibration problems. 

\end{abstract}



\section{Introduction}

In Bayesian experimental design (BED,~\cite{chaloner_bayesian_1995, lindley_measure_1956, lindley_bayesian_1972}), the expected information gain (EIG) is commonly defined as an outer expectation of a posterior-to-prior information measure with respect to the joint distribution of parameters and hypothetical observations~\cite{jones_bayes_2016, ryan_review_2016}. Estimating EIG or its gradient for design optimization therefore requires handling a large and heterogeneous family of posterior distributions, which is a central challenge in BED and has motivated extensive prior work~\cite{rainforth_modern_2023, huan_optimal_2024}. A straightforward way to obtain high-fidelity EIG estimates is to perform inference separately for each outer sample of parameters and data. Classical approaches following this principle include nested Monte Carlo estimators~\cite{ryan2003estimating}, particle-based methods such as sequential Monte Carlo estimators~\cite{drovandi2014sequential, grosse2016measuring}, or local Laplace approximations~\cite{long2013fast, beck2018fast, englezou2022approximate}. These methods can already be computationally demanding for a single outer sample: Monte Carlo-based estimators require repeated forward-model evaluations, while Laplace approximations involve solving an optimization problem and computing the Hessian or its approximation. Moreover, this inner inference must be repeated across many outer samples, which makes the overall computation costly. Even when the cost of inner inference for each outer sample can be reduced~\cite{YANG2026114469, callahanreverse}, the total forward-model cost often remains substantial because it still grows with both the number of outer samples and the cost of the inner inference required for each one.

To improve efficiency at a higher level, the inner inference can be amortized across outer samples rather than performed independently for each outer sample~\cite{foster_deep_2021, huang2024amortized}. This is reasonable considering the posteriors induced by different outer samples may be similar, so that some shared inference mechanism can be effective across multiple outer samples. This idea can be implemented in different ways. For example, one may train a shared variational surrogate that produces an approximate posterior for each outer sample~\cite{foster_unified_2020, dong2025variational}, or construct a common proposal distribution and use its weighted samples to approximate the per-outer-sample (individual) posterior~\cite{ryan2003estimating, huan_simulation-based_2013}. By replacing repeated per-outer-sample inner inference with a shared surrogate or proposal, such methods can substantially reduce computational cost. Their effectiveness, however, depends on whether the shared inference mechanism remains sufficiently accurate across the full range of outer samples. For example, if the proposal poorly matches some individual posteriors, the resulting posterior approximations or importance-sampling estimates may deteriorate~\cite{parks2006recycling, feng2019layered}.

This motivates the development of stronger shared inference that remains accurate across a heterogeneous family of individual posteriors. For importance sampling in particular, the key is to construct a shared proposal that better matches the family of individual posteriors. Recent work has introduced the geometric pooled posterior (GPP) for this purpose~\cite{iollo2024bayesian}. Formally, the GPP is defined as the distribution that minimizes an average Kullback-Leibler (KL) divergence to all the individual posteriors, which yields a logarithmic pool, equivalently a geometric mixture of those posteriors. Because it is, on average, closer to the individual posteriors, it is expected to provide better overlap with the target posteriors and therefore improve the effectiveness of importance sampling. Despite its name, however, the GPP is not a posterior distribution of the original BED problem, but a proposal distribution for importance sampling. It is not necessary to explicitly compute all individual posteriors and then combine them. Instead, the same distribution can be expressed directly through the prior and pooled likelihood terms across outer samples. Prior work has generated samples from the GPP using conditional diffusion models and demonstrated the effectiveness of this strategy~\cite{iollo2024bayesian}. However, the GPP still represents a single global proposal for the entire outer ensemble, i.e., a "one-size-fits-all" strategy. In this way, GPP may adapt well to the dominant region of the posterior family, yet remain poorly matched to outlying individual posteriors. When this happens, the overlap may be limited, and the resulting importance-sampling estimates may degrade~\cite{parks2006recycling, feng2019layered}, e.g., weight degeneracy and unstable estimates.

To address this limitation, we propose a grouping strategy over outer samples. Instead of forcing a single global proposal on the entire outer ensemble, we partition the outer samples into subsets whose individual posteriors exhibit a higher degree of similarity, compared to the similarity across the outer ensemble as a whole. We then construct a dedicated geometric pooled posterior for each group, and then use that group-specific proposal to perform importance sampling for the outer samples within the group. Because the samples in each group are organized to have similar posteriors, the resulting group-specific proposal can achieve substantially better proposal-target overlap for the corresponding individual posteriors than a single global proposal. At the same time, allowing different groups to have different proposals enables the approximation to adapt to distinct regions of the posterior family, rather than fitting all individual posteriors with one shared distribution. In this way, the proposed strategy improves the stability and reliability of importance-sampling estimation under posterior heterogeneity. Importantly, it still retains the main advantage of amortization: inference is shared within each group, so one avoids the cost of repeating a fully separate inner inference for every outer sample.

While this grouping strategy improves the importance-sampling quality, it also extends proposal-sample generation from a single global proposal to multiple group-specific proposals. Proposal-sample generation is necessary because importance sampling relies on samples drawn from the proposal to represent each target distribution. As a result, moving from one proposal to several group-specific proposals introduces additional sampling cost. Under standard likelihood-based inference methods such as MCMC, SMC, or recent conditional diffusion samplers~\cite{iollo2024bayesian}, generating samples for even a single proposal can already be expensive, since each group-specific proposal typically requires repeated forward-model evaluations. Extending this process to multiple groups therefore replicates the cost across groups, causing the total computational cost to grow roughly in proportion to the number of groups. This makes traditional sample-generation methods poorly suited to the grouped setting.

To avoid this additional computational burden, we adapt ensemble Kalman inversion (EKI,~\cite{kovachki_ensemble_2019, ding2021ensemble}) to generate samples from grouped geometric pooled posteriors. A key advantage of this approach is that the forward-model evaluation cost does not grow with the number of groups: generating multiple group-specific proposals costs essentially the same as generating a single global proposal. This is because forward-model evaluations are required only in the prediction step of EKI, while the subsequent group-specific updates are carried out entirely in the analysis step. As a result, additional group-specific pooled posterior samples can be obtained with no extra forward-model cost, giving our approach a substantial computational advantage over traditional sampling methods. 

While the use of EKI to generate approximate posterior samples in BED problems has been explored in prior work~\cite{YANG2026114469,yang_active_2025}, our contribution is to adapt it to grouped pooled posteriors and use the resulting ensembles as proposals for importance sampling, rather than as direct approximations for information-gain estimation. Building on this idea, we further develop practical implementations for pooled posteriors and introduce a conservative ESS-based diagnostic that helps assess importance-sampling quality and guide grouping in a self-consistent and adaptive manner.

We apply our method to find the best data for both physical parameter inference and model discrepancy calibration, where a high-dimensional neural network is introduced to correct a governing partial differential equation (PDE)~\cite{levine2022framework, ebers2024discrepancy}. This is a challenging setting for BED because model discrepancy can bias forward evaluations~\cite{rainforth_modern_2023, barlas2025robust, kennedy_bayesian_2001,feng_optimal_2015}, while the resulting discrepancy network often leads to a high-dimensional calibration problem~\cite{dong2025stochastic, wu2024learning, yang_active_2025}. As a result, performing Bayesian updates and estimating information gains in that high-dimensional space becomes prohibitively expensive~\cite{alexanderian2021optimal,wu2023fast,go2025sequential, neuberger2025goal}. We build on the decoupled framework~\cite{YANG2026114469}, which treats the physical and network parameters separately, and replace its network-side gradient estimator AD-EKI with our grouped pooled posterior method. Also comparing with the conditional diffusion method~\cite{iollo2024bayesian}, we show that this modification achieves comparable performance in design search and parameter estimation at a lower computational cost.


The key highlights of our work are as follows.
\begin{itemize}
\item We propose a grouped geometric pooled posterior method for Bayesian experimental design, which partitions outer samples into groups and performs importance sampling within each group. This improves the quality of importance-sampling estimation while preserving the computational advantage of amortization.
\item We adapt ensemble Kalman inversion to make sample generation under grouping setting practical: compared with using a single global pooled proposal, generating samples for multiple groups does not increase forward-model evaluation cost. We also introduce a conservative ESS-based diagnostic to guide the grouping.
\item We validate our method on model-discrepancy calibration across parametric and structural error types, demonstrating strong performance and marked computational savings, especially in neural-network-correction scenarios, relative to benchmarks in high-dimensional settings.
\end{itemize}

This paper is organized as follows. Section \ref{sec: problem} formulates the problem and introduces the pooled-posterior-based solution. Section \ref{sec: eki} presents the tailored ensemble method for drawing pooled posterior samples. Section \ref{sec: group} introduces the grouping strategy and the full grouped framework. Section \ref{sec: model discrepancy} describes the application to model discrepancy calibration. Numerical results are reported in Section \ref{sec: Numerical Results}, and conclusions are drawn in Section \ref{Conclusion}.


\section{Problem formulation and EIG gradient estimation}\label{sec: problem}

\subsection{Bayesian experimental design}

The Bayesian experimental design~\cite{lindley_bayesian_1972, rainforth_modern_2023,huan_optimal_2024} provides a general framework to systematically seek the optimal design by solving the optimization problem: 
\begin{equation}
\label{eq:oed_opt}
    \begin{aligned}
        \mathbf{d}^* &= \argmax_{\mathbf{d}\in\mathcal{D}} \mathbb{E}[U(\boldsymbol{\theta},\mathbf{y},\mathbf{d})]\\
        &= \argmax_{\mathbf{d}\in\mathcal{D}} \int_\mathcal{Y} \int_{\boldsymbol{\Theta}} U(\boldsymbol{\theta},\mathbf{y},\mathbf{d}) p(\boldsymbol{\theta},\mathbf{y}|\mathbf{d}) \mathrm{d}\boldsymbol{\theta}\mathrm{d}\mathbf{y},
    \end{aligned}
\end{equation}
where $\mathbf{y} \in \mathcal{Y} \subset \mathbb{R}^{d_{\mathbf{y}}}$ is data from the experimental design $\mathbf{d} \in \mathcal{D}\subset \mathbb{R}^{d_{\mathbf{d}}}$, $\boldsymbol{\theta} \in \boldsymbol{\Theta}\subset \mathbb{R}^{d_{\boldsymbol{\theta}}}$ denotes the target parameters, and $U$ is the utility function taking the inputs of $\mathbf{y}$ and $\boldsymbol{\theta}$ and returning a real value, which reflects the specific purpose of designing the experiment. The optimal design $\mathbf{d}^*$ is obtained by maximizing the expected utility function $\mathbb{E}[U(\boldsymbol{\theta},\mathbf{y},\mathbf{d})]$ over the design space $\mathcal{D}$. The term $p(\boldsymbol{\theta},\mathbf{y}|\mathbf{d})$ is the joint conditional distribution of data and parameters.

The utility function $U(\boldsymbol{\theta},\mathbf{y},\mathbf{d})$ can be regarded as the information gain obtained from the data $\mathbf{y}$ for the corresponding design $\mathbf{d}$. A most common choice in Bayesian setting is the Kullback-Leibler (KL) divergence between the posterior and the prior distribution of parameters $\boldsymbol{\theta}$:
\begin{equation}
    \begin{aligned}
        U(\boldsymbol{\theta},\mathbf{y},\mathbf{d}) &= D_{\textrm{KL}} \bigl(p(\boldsymbol\theta|\mathbf{y},\mathbf{d})\| p(\boldsymbol\theta)\bigr)\\
        &=\mathbb{E}_{\boldsymbol\theta|\mathbf{d},\mathbf{y}}(\log p(\boldsymbol{\theta}|\mathbf{y},\mathbf{d}) - \log p(\boldsymbol{\theta}))\\
        &=\int_{\boldsymbol{\Theta}} p(\boldsymbol{\theta}|\mathbf{y},\mathbf{d}) \log(\frac{p(\boldsymbol{\theta}|\mathbf{y},\mathbf{d})}{p(\boldsymbol{\theta})})\mathrm{d}\boldsymbol{\theta},
    \end{aligned}
    \label{eq: kld utility}
\end{equation}
where $p(\boldsymbol\theta|\mathbf{y},\mathbf{d})$ is the posterior distribution of $\boldsymbol\theta$ given a design $\mathbf{d}$ and the data $\mathbf{y}$. Note that data $\mathbf{y}$ could be either from the actual experimental measurement or the numerical model simulation. The use of KL divergence as the utility function leads to the definition of expected information gain (EIG,~\cite{lindley_measure_1956}) that integrates the information gain over all possible predicted data:
\begin{equation}
\label{eq:EIG}
    \begin{aligned}
        I (\mathbf{d}) &=\mathbb{E}_{\mathbf{y} | \mathbf{d}} [D_{\textrm{KL}}(p(\boldsymbol\theta|\mathbf{y},\mathbf{d})\| p(\boldsymbol\theta))]\\
  &=\int_{\mathcal{Y}}D_{\textrm{KL}}(p(\boldsymbol\theta|\mathbf{y},\mathbf{d})\| p(\boldsymbol\theta)) p(\mathbf{y}|\mathbf{d})\mathrm{d}\mathbf{y},\\      
    \end{aligned}
\end{equation}
where $I$ represents the EIG, and $p(\mathbf{y}|\mathbf{d}):=\mathbb{E}_{\boldsymbol{\theta}}[p(\mathbf{y}|\boldsymbol{\theta},\mathbf{d})]$ is the distribution of the predicted data among all possible parameter values given a certain design.

\subsection{Importance sampling with geometric pooled posterior}

To solve the design optimization problem, importance sampling is one of the most common methods to estimate the gradient of EIG w.r.t. design~\cite{goda2022unbiased, ao2024estimating}:
\begin{equation}
    \nabla_{\mathbf{d}} I(\mathbf{d}) = \mathbb{E}_{p(\boldsymbol\theta,\mathbf{y}|\mathbf{d})} \left[ \nabla_{\mathbf{d}}\log (\mathbf{d},\mathbf{y},\boldsymbol\theta) - \mathbb{E}_{q(\boldsymbol\theta',\mathbf{y}|\mathbf{d})} \left[ \frac{p(\boldsymbol\theta' | \mathbf{y},\mathbf{d})}{q(\boldsymbol\theta',\mathbf{y}|\mathbf{d})} \nabla_{\mathbf{d}}\log (\mathbf{d},\mathbf{y},\boldsymbol\theta,\boldsymbol\theta')\right] \right]
\end{equation}
where $\boldsymbol\theta'$ also represents the unknown parameter but requires another sampling process, which is independent to $\boldsymbol\theta$. $q$ is the proposal distribution for importance sampling, which needs to be chosen by users. A good choice of proposal distribution can significantly improve the efficiency.

The most common choice for the proposal distribution is the prior distribution. However, this simple choice is very likely to cause weight degeneracy. Previous work~\cite{iollo2024bayesian} proposes using a geometric pooled posterior (GPP) as the proposal distribution in importance sampling. Let the prior be \(p(\boldsymbol\theta)\), samples $\{\boldsymbol\theta_i\}_{i=1}^N$ ($\boldsymbol\theta_i \in\mathbb{R}^{{d}_{\boldsymbol\theta}}$) and let the outer samples be \(Y=\{\mathbf{y}_i\}_{i=1}^{N}\) with \(\mathbf{y}_i \in \mathbb{R}^{d_\mathbf{y}}\). The outer samples \(Y=\{\mathbf{y}_i\}_{i=1}^{N}\) is generated with prior samples through the forward model $g(\boldsymbol\theta)$ and with observation noise $\boldsymbol\epsilon_i \sim \mathcal{N}(\boldsymbol{0},\boldsymbol\Sigma)$ : 
\begin{equation}
    \mathbf{y}_i=f(\boldsymbol\theta_i)+\boldsymbol\epsilon_i
    \label{eq:forward}
\end{equation}

 With nonnegative weights \(\{\nu_i\}_{i=1}^{N}\) such that \(\sum_{i=1}^{N}\nu_i=1\), the geometric pooled posterior is defined as
\begin{equation}
\label{eq:gpp-def}
p_{\mathrm{GPP}}(\boldsymbol\theta|{Y})
\;\propto\;
p(\boldsymbol\theta)\,\prod_{i=1}^{N} p(\mathbf{y}_i|\boldsymbol\theta)^{\nu_i}.
\end{equation}

The sampling-based estimation of the gradient of EIG w.r.t design is written as
\begin{equation}\label{eq: mc eig gradient to design}
    \nabla_{\mathbf{d}} I(\mathbf{d}) \approx \frac{1}{N} \sum_{i=1}^N \left[ \nabla_{\mathbf{d}}\log (\mathbf{d},\mathbf{y}_i,\boldsymbol\theta_i) - \frac{1}{M} \sum_{j=1}^M \left[ \omega_j \nabla_{\mathbf{d}}\log (\mathbf{d},\mathbf{y}_i,\boldsymbol\theta_i,\boldsymbol\theta'_j)\right] \right]
\end{equation}
where $\omega_j$ is the importance sampling weight with GPP being the proposal,
\begin{equation}
    \omega_j =\frac{p(\boldsymbol\theta'_j | \mathbf{y}_i,\mathbf{d})}{p_\text{GPP}(\boldsymbol\theta'_j|Y,\mathbf{d})}.
\end{equation}
For non-linear or complex systems, it should be noted that the exact value of both probability density in this weight is often not available. Then self normalized importance sampling weights:
\begin{equation}
\begin{aligned}
    \omega_j 
    &= \frac{p(\boldsymbol\theta'_j | \mathbf{y}_i,\mathbf{d})}
            {p_\text{GPP}(\boldsymbol\theta'_j| Y,\mathbf{d})}\\[4pt]
    &= 
    \frac{
        p(\boldsymbol\theta'_j)\, p(\mathbf{y}_i|\boldsymbol\theta'_j ,\mathbf{d}) 
        \,/\, p(\mathbf{y}_i|\mathbf{d})
    }{
        p(\boldsymbol\theta'_j)\, \prod_{k=1}^{N} p(\mathbf{y}_k| \boldsymbol\theta'_j,\mathbf{d})^{\nu_k}
        \,/\, p(Y|\mathbf{d})
    }\\[4pt]
    &\propto 
    \frac{
        p(\mathbf{y}_i|\boldsymbol\theta'_j ,\mathbf{d})
    }{
        \prod_{k=1}^{N} p(\mathbf{y}_k| \boldsymbol\theta'_j,\mathbf{d})^{\nu_k}
    }.
\end{aligned}
\label{eq:snis}
\end{equation}

With the Gaussian observation noise defined in Eq.~\eqref{eq:forward}, the likelihood terms in the last line of Eq.~\eqref{eq:snis} can be directly written out in closed-form. Another benefit of this formula is that the prior probability density is canceled, which means this formula works even without knowing explicit form of prior distribution, e.g., only the prior samples are provided.

\section{Sampling with ensemble Kalman inversion}\label{sec: eki}

Generation of samples from the geometric pooled posterior (GPP) can be formulated within an ensemble Kalman framework. Specifically, we employ the ensemble Kalman inversion (EKI) as an approximate Bayesian update to construct GPP proposals. For given observation data, we draw samples from the prior and apply a one-step EKI update to map these prior samples to posterior samples. Under the assumption of linear model and Gaussian noise, a one-step (stochastic) EKI update coincides with the exact Bayesian posterior update and thus provides the conceptual foundation of our method. Similar approaches have already been explored in a variety of studies~\cite{yang_active_2025, YANG2026114469}. However, to tailor EKI to the GPP setting considered here, several modifications are required.

\subsection{Naive stacking version} Define the stacked observation and repeated prediction
\begin{equation}
\label{eq:stack-Y}
Y
=
\begin{bmatrix}
\mathbf{y}_1 \\ \mathbf{y}_2 \\ \vdots \\ \mathbf{y}_N
\end{bmatrix}
\in \mathbb{R}^{Nd_{\mathbf{y}}},
\qquad
F(\boldsymbol{\theta})
=
\begin{bmatrix}
f(\boldsymbol{\theta}) \\ f(\boldsymbol{\theta}) \\ \vdots \\ f(\boldsymbol{\theta})
\end{bmatrix}
\in \mathbb{R}^{Nd_{\mathbf{y}}}.
\end{equation}
where we use $Y$ as the stacked outer samples.

Assume Gaussian observation models $p(\mathbf{y}_i| \boldsymbol{\theta})=\mathcal{N}\big(\mathbf{y}_i; f(\boldsymbol\theta),\boldsymbol\Sigma_i\big)$ with positive weights $\{\nu_i\}_{i=1}^N$ satisfying $\sum_i \nu_i=1$. To incorporate the geometric weights, assign blockwise noise covariances $\boldsymbol\Sigma_i/\nu_i$ for the $i$-th block:
\begin{equation}
\label{eq:sigma-gpp}
\boldsymbol\Sigma_{\mathrm{GPP}}
=
\mathrm{diag}\!\left(\tfrac{\boldsymbol\Sigma_1}{\nu_1},\,\tfrac{\boldsymbol\Sigma_2}{\nu_2},\,\ldots,\,\tfrac{\boldsymbol\Sigma_N}{\nu_N}\right)
\in \mathbb{R}^{N{ d_{\mathbf{y}}} \times N{ d_{\mathbf{y}}}}.
\end{equation}

Let $\{\boldsymbol\theta^{(j)}\}_{j=1}^J$ be an ensemble draw from the prior, with $J$ denoting the ensemble size. Define sample means $\bar{\boldsymbol\theta}=\frac{1}{J}\sum_j \boldsymbol\theta^{(j)}$ and $\bar F=\frac{1}{J}\sum_j F(\boldsymbol\theta^{(j)})$, and covariances
\begin{align}
\label{eq:PthF}
P_{\boldsymbol\theta F}
&=
\frac{1}{J-1}\sum_{j=1}^{J}
\big(\boldsymbol\theta^{(j)}-\bar{\boldsymbol\theta}\big)\big(F(\boldsymbol\theta^{(j)})-\bar F\big)^{\top},\\
\label{eq:PFF}
P_{FF}
&=
\frac{1}{J-1}\sum_{j=1}^{J}
\big(F(\boldsymbol\theta^{(j)})-\bar F\big)\big(F(\boldsymbol\theta^{(j)})-\bar F\big)^{\top}.
\end{align}
The Kalman gain for the augmented system is
\begin{equation}
\label{eq:kalman-gain}
K
=
P_{\boldsymbol\theta F}\,\big(P_{FF}+\Sigma_{\mathrm{GPP}}\big)^{-1},
\end{equation}
and each ensemble member is updated by
\begin{equation}
\label{eq:eki-update}
\boldsymbol\theta^{(j)}_{\mathrm{new}}
=
\boldsymbol\theta^{(j)}
+
K\,\big( Y - F(\boldsymbol\theta^{(j)}) \big),
\qquad j=1,\ldots,J.
\end{equation}

Perturbations for $F(\boldsymbol\theta^{(j)})$ in Eq.~\eqref{eq:eki-update} are required in practice to preserve sampling variability, as called stochastic EKI formula. The updated ensemble $\{\boldsymbol\theta^{(j)}_{\mathrm{new}}\}_{j=1}^J$ approximates the desired GPP distribution and can be directly used for the subsequent importance sampling. In this case the ensemble size $J$ equals to the inner sample size $M$. The proposed EKI-based sample generation method is generally more computationally feasible than other sampling methods, i.e., MCMC, conditional diffusion, etc.

To justify this construction, we first establish an exact equivalence between the likelihood of GPP and a Gaussian likelihood under a stacked-observation model, and then relate the one-step stochastic EKI update to the induced posterior, which becomes exact under linear-Gaussian assumptions as ensemble size approaches infinity.

\begin{proposition}[Stacked-observation representation of geometric pooling]
\label{prop:stacked_gpp_equivalence}\\
Let $Y$ and $F(\boldsymbol{\theta})$ be defined in \eqref{eq:stack-Y} and let $\boldsymbol\Sigma_{\mathrm{GPP}}$ be defined in \eqref{eq:sigma-gpp}.
Under the Gaussian observation models $p(\mathbf y_i\mid\boldsymbol{\theta})=\mathcal N(\mathbf y_i; f(\boldsymbol{\theta}),\boldsymbol\Sigma_i)$ with weights
$\{\nu_i\}_{i=1}^N$ satisfying $\nu_i>0$ and $\sum_i \nu_i=1$, we have
\begin{equation}
\label{eq:stacked_like_equiv}
\begin{aligned}
\prod_{i=1}^N p(\mathbf y_i\mid\boldsymbol{\theta})^{\nu_i}
\;&\propto \; \exp\!\left(
-\tfrac12 \sum_{i=1}^N \nu_i \, \|\mathbf{y}_i - f(\boldsymbol{\theta})\|^2_{\boldsymbol\Sigma_i}
\right) \\
&\propto \; 
\mathcal N\!\big(Y;\, F(\boldsymbol{\theta}),\, \Sigma_{\mathrm{GPP}}\big)
\end{aligned}
\end{equation}
where the proportionality constant is independent of $\boldsymbol{\theta}$, and $||a||_{\Sigma}:=a^\top \Sigma^{-1} a$.
\end{proposition}

\begin{proof}

We show that the likelihood in the desired geometric pooled posterior is proportional to a stacked-observation Gaussian density, i.e., the second proportionality in Eq.~\eqref{eq:stacked_like_equiv}.

Let $Y$ and $F(\boldsymbol{\theta})$ be defined in \eqref{eq:stack-Y}, and let
$\boldsymbol\Sigma_{\mathrm{GPP}}=\mathrm{diag}(\boldsymbol\Sigma_1/\nu_1,\ldots,\boldsymbol\Sigma_N/\nu_N)$ as in \eqref{eq:sigma-gpp}.
The multivariate Gaussian density is written as
\[
\mathcal N\!\big(Y;F(\boldsymbol{\theta}),\boldsymbol\Sigma_{\mathrm{GPP}}\big)
=
C_1\,
\exp\!\Big(-\tfrac12\|Y-F(\boldsymbol{\theta})\|_{\boldsymbol\Sigma_{\mathrm{GPP}}}^2\Big),
\]
where $C_1=(2\pi)^{-Nd_\mathbf{y}/2}|\boldsymbol\Sigma_{\mathrm{GPP}}|^{-1/2}$ is independent of $\boldsymbol{\theta}$.

Since $\boldsymbol\Sigma_{\mathrm{GPP}}^{-1}=\mathrm{diag}(\nu_1\boldsymbol\Sigma_1^{-1},\ldots,\nu_N\boldsymbol\Sigma_N^{-1})$ and
$Y-F(\boldsymbol{\theta})=[\mathbf y_1-f(\boldsymbol{\theta});\ldots;\mathbf y_N-f(\boldsymbol{\theta})]\in \mathbb{R}^{Nd_{\mathbf{y}}}$, we have
\[
\begin{aligned}
    \|Y-F(\boldsymbol{\theta})\|_{\boldsymbol\Sigma_{\mathrm{GPP}}}^2
&= (Y-F(\boldsymbol{\theta}))^\top {\boldsymbol\Sigma_{\mathrm{GPP}}}^{-1} (Y-F(\boldsymbol{\theta}))\\
&= [\mathbf y_1-f(\boldsymbol{\theta});\ldots;\mathbf y_N-f(\boldsymbol{\theta})]^\top \left[\begin{matrix}
    \nu_1 \Sigma_1^{-1} (\mathbf y_1-f(\boldsymbol{\theta})) \\
    \cdots\\
    \nu_N \Sigma_N^{-1} (\mathbf y_N-f(\boldsymbol{\theta})) \\
\end{matrix}\right]\\
&=\sum_{i=1}^N (\mathbf y_i-f(\boldsymbol{\theta}))^\top(\nu_i\boldsymbol\Sigma_i^{-1})(\mathbf y_i-f(\boldsymbol{\theta}))\\
&=
\sum_{i=1}^N \nu_i\|\mathbf y_i-f(\boldsymbol{\theta})\|_{\boldsymbol\Sigma_i}^2.
\end{aligned}
\]
Therefore,
\[
\mathcal N\!\big(Y;F(\boldsymbol{\theta}),\boldsymbol\Sigma_{\mathrm{GPP}}\big)
\;\propto\;
\exp\!\Big(-\tfrac12\sum_{i=1}^N \nu_i\|\mathbf y_i-f(\boldsymbol{\theta})\|_{\boldsymbol\Sigma_i}^2\Big).
\]

\end{proof}

\begin{lemma}[Asymptotic exactness of one-step stochastic EKI]
\label{lem:eki_linear_gaussian}
Consider the linear--Gaussian model $Y = A\boldsymbol{\theta} + \varepsilon$ with prior $\boldsymbol{\theta}\sim\mathcal N(m_0,C_0)$ and noise
$\varepsilon\sim\mathcal N(0,\boldsymbol\Sigma)$.
In the limit $J\to\infty$, one step of the stochastic EKI update produces an updated ensemble whose empirical distribution converges to the exact Bayesian posterior $p(\boldsymbol{\theta}\mid Y)$.

In our stacked construction, the associated posterior is written as
\[
\{\boldsymbol{\theta}^{(j)}_\mathrm{new}\}_{j=1}^J \sim p(\boldsymbol{\theta}\mid Y)\ \propto\ p(\boldsymbol{\theta})\,p(Y\mid \boldsymbol{\theta},\boldsymbol\Sigma_{\mathrm{GPP}}),
\]
with
\[
p(Y\mid \boldsymbol{\theta},\boldsymbol\Sigma_{\mathrm{GPP}})=\mathcal N\!\big(Y;F(\boldsymbol{\theta}),\boldsymbol\Sigma_{\mathrm{GPP}}\big).
\]

\end{lemma}

\begin{proposition}[Stacked-observation EKI targets the GPP posterior in the linear--Gaussian limit]
\label{prop:eki_targets_gpp}
Under the linear--Gaussian assumption and as $J\to\infty$, one step of the stochastic stacked-observation EKI update targets the Bayesian posterior associated with the prior $p(\boldsymbol{\theta})$ and the stacked likelihood $\mathcal N(Y;F(\boldsymbol{\theta}),\boldsymbol\Sigma_{\mathrm{GPP}})$ (Lemma~\ref{lem:eki_linear_gaussian}). 
By Proposition~\ref{prop:stacked_gpp_equivalence}, this posterior coincides with the desired geometric pooled posterior induced by the pooled likelihood $\prod_{i=1}^N p(\mathbf y_i\mid\boldsymbol{\theta})^{\nu_i}$:
\[
\{\boldsymbol{\theta}^{(j)}_\mathrm{new}\}_{j=1}^J \sim p(\boldsymbol{\theta})\,\mathcal N\!\big(Y;F(\boldsymbol{\theta}),\boldsymbol\Sigma_{\mathrm{GPP}}\big)
\;\propto\;
p(\boldsymbol{\theta})\prod_{i=1}^N p(\mathbf y_i\mid\boldsymbol{\theta})^{\nu_i}\;\propto\;p_\mathrm{GPP}.
\]
\end{proposition}

In the naive stacked construction, computing the Kalman gain could be expensive because it involves covariance matrix inversion in a high-dimensional stacked observation space. To avoid stacking, we derive an equivalent mean-observation formulation.

\subsection{Mean-observation formulation} Define the effective precision and mean outer samples
\begin{equation}
\label{eq:mean-params}
\boldsymbol\Sigma_\nu^{-1} \;=\; \sum_{i=1}^N \nu_i \,\boldsymbol\Sigma_i^{-1},
\qquad
\bar{\mathbf{y}}_\nu \;=\; \boldsymbol\Sigma_\nu \,\Big(\sum_{i=1}^N \nu_i \,\boldsymbol\Sigma_i^{-1} \mathbf{y}_i\Big).
\end{equation}

An EKI implementation without stacking all outer samples $\mathbf{y}_i$ and without increasing the dimension of the noise covariance matrix is given by:
\begin{equation}
\boldsymbol\theta^{(j)}_{\mathrm{new}}
=
\boldsymbol\theta^{(j)} + K_\nu\big(\bar{\mathbf{y}}_\nu - f(\boldsymbol\theta^{(j)})\big),
\qquad
K_\nu = P_{\boldsymbol\theta f}\big(P_{ff}+\boldsymbol\Sigma_\nu\big)^{-1}.
\end{equation}
This updated ensemble also approximates the desired GPP distribution.

\begin{proposition}[Mean-observation representation of geometric pooling]
\label{prop:mean_gpp_equivalence}
Under the Gaussian observation models $p(\mathbf y_i\mid\theta)=\mathcal N(\mathbf y_i; f(\theta),\boldsymbol\Sigma_i)$
with weights $\nu_i>0$ and $\sum_i \nu_i=1$, we have
\begin{equation}
\label{eq:mean-like}
\prod_{i=1}^N p(\mathbf{y}_i|\boldsymbol\theta)^{\nu_i}
\;\propto\;
\exp\!\left(
-\tfrac12 \, \|\, \bar{\mathbf{y}}_\nu - f(\boldsymbol\theta)\,\|^2_{\boldsymbol\Sigma_\nu}
\right)
\;\propto\;
\mathcal{N}\!\big(\bar{\mathbf{y}}_\nu;\, f(\boldsymbol\theta),\, \boldsymbol\Sigma_\nu\big),
\end{equation}
where the proportionality constant is independent of $\theta$. Therefore, the GPP posterior can be written as
\[
p_{\mathrm{GPP}}(\boldsymbol\theta)\;\propto\;
p(\boldsymbol\theta)\,\mathcal{N}\!\big(\bar{\mathbf{y}}_\nu;\, f(\boldsymbol\theta),\, \boldsymbol\Sigma_\nu\big).
\]
That is, the GPP is exactly equivalent to a single weighted mean observation $\bar{\mathbf{y}}_\nu$ with effective noise covariance $\boldsymbol\Sigma_\nu$.
\end{proposition}

The proof of Proposition~\ref{prop:mean_gpp_equivalence} is given in Appendix~\ref{apd: proof of proportion}.

\subsection{Homoskedastic case in BED}

It should be noted that in the context of BED, all outer samples  $\mathbf{y}_i$ come from the same observation model, which means \(\boldsymbol\Sigma_i \equiv \boldsymbol\Sigma\) holds for all \(i=1,2,\cdots,N\). Hence \(\boldsymbol\Sigma_\nu=\boldsymbol\Sigma\), and
\begin{equation}
p_{\mathrm{GPP}}(\boldsymbol\theta)\;\propto\;
p(\boldsymbol\theta)\,\mathcal{N}\!\Big(\bar{\mathbf{y}}_\nu;\, f(\boldsymbol\theta),\, \boldsymbol\Sigma\Big),
\end{equation}
so for BED problems, one may \emph{'average' the outer samples and use the original noise covariance $\boldsymbol\Sigma$ in a standard EKI without further modification}:

\begin{equation}\label{eq:eki final}
\boldsymbol\theta^{(j)}_{\mathrm{new}}
=
\boldsymbol\theta^{(j)} + K_\nu\big(\bar{\mathbf{y}}_\nu - f(\boldsymbol\theta^{(j)})\big),
\qquad
K_\nu = P_{\boldsymbol\theta f}\big(P_{ff}+\boldsymbol\Sigma\big)^{-1}.
\end{equation}

This is algebraically equivalent to the augmented-stacking construction with block covariances, but is often simpler and more computational friendly. We mark the new ensemble as $\{\theta^{(j)}|Y\}_{j=1}^J$, indicating that it is obtained by an EKI update driven by the pooled likelihood defined on $Y$.

\section{Grouped geometric pooled posterior}\label{sec: group}

Although using GPP as a proposal is better than the prior, it is still possible that, for outer samples in the tail of the distribution, the gradients $\partial I/\partial \mathbf{d}$ become extreme values and degrade the overall estimation. To address this issue, we leverage the strengths of the EKI sampling process and propose to group the outer samples. For each group of outer samples, we construct a group-specific pooled posterior and use it as the proposal distribution in importance sampling restricted to that group. Since each grouped pooled posterior is tailored to a subset of outer samples, this strategy reduces the importance sampling variance relative to using a single global proposal.

The proposed framework has two key advantages. First, generating the GPP samples reuses the EKI updating process and therefore does not require any additional forward PDE simulations beyond those already carried out during global EKI. Second, the intermediate ensemble statistics produced by EKI yield a computable conservative diagnostic on the effective sample size (ESS), which we use to guide the construction and refinement of the groups. On the other hand, employing multiple grouped proposals inevitably increases the number of forward PDE simulations, because for each group we must propagate the associated GPP samples through the forward model to obtain their predictions. Nevertheless, the resulting variance reduction is substantially more effective than simply increasing the number of inner samples under a single proposal. Overall, this allows us to achieve a given accuracy level with fewer inner samples per group while keeping the overall computational budget under control.

\textbf{Conservative diagnostic on effective sample size}: To quantify the sampling efficiency with GPP proposal, we derive an explicit approximation for the ESS under a Gaussian observation model. Assuming the importance weights follow a log-normal distribution, the ESS for a specific outer sample $\mathbf{y}_i$ is given by:
\begin{equation}
  \label{eq:ess-bound}
  \mathrm{ESS}(\mathbf{y}_i)
  \approx
  J \exp\!\left(
    -\,(\mathbf{y}_i -\bar{\mathbf{y}}_\nu)^\top R^{-1}\,\Sigma_{ff}\,R^{-1}(\mathbf{y}_i -\bar{\mathbf{y}}_\nu)
  \right),
\end{equation}
where $J$ is the number of inner samples, $\bar{\mathbf{y}}_\nu$ is the mean observation in the group, $R$ is the observation-noise covariance (that is, $\Sigma$), and $\Sigma_{ff}$ denotes the predictive covariance of the forward-model outputs under the proposal distribution. This formula captures the distributional mismatch between the target and the proposal by measuring the Mahalanobis distance between their respective conditioning observations, $\mathbf{y}_i$ and $\bar{\mathbf{y}}_\nu$. The derivation of this formula can be found in Appendix~\ref{apd: ess}. Once the outer samples and the observation operator are fixed, both $\mathbf{y}_i$ and $\bar{\mathbf{y}}_\nu$ are directly available, so evaluating Eq. \eqref{eq:ess-bound} only requires an estimate of $\Sigma_{ff}$.

Computing $\Sigma_{ff}$ directly would require propagating the GPP samples through the forward PDE model to obtain their predictions and then forming the corresponding sample covariance, which incurs an additional set of PDE solves. Notice that the EKI prediction step produces an ensemble of forward-model outputs from the forecast (prior) ensemble, and their sample covariance yields a readily available quantity $P_{ff}$. Since $P_{ff}$ is a predictive covariance in the same observation space as $\Sigma_{ff}$, we substitute $P_{ff}$ for $\Sigma_{ff}$ in Eq. \eqref{eq:ess-bound} to obtain a computable ESS estimate. 

This substitution is conservative: intuitively, the GPP incorporates information from the outer samples and is therefore more concentrated than the prior, which in turn reduces the predictive uncertainty in the observation space. Consequently, the predictive covariance under the GPP ensemble is no larger than that under the forecast ensemble, i.e.,
\begin{equation}
    \Sigma_{ff} \;\preceq\; P_{ff},
\end{equation}
where the $\preceq$ denotes the Loewner (positive semidefinite) order. This covariance contraction holds rigorously under the linear-Gaussian assumptions; a proof can be found in the appendix~\ref{apd: lower covariance}.

Substituting the larger prior-based quantity $P_{ff}$ into Eq. \eqref{eq:ess-bound}:
\begin{equation}
  \label{eq:hat-ess-bound}
  \hat{\mathrm{ESS}}(\mathbf{y}_i)
  =
  J \exp\!\left(
    -\,(\mathbf{y}_i - \bar{\mathbf{y}}_\nu)^\top R^{-1}\,P_{ff}\,R^{-1}(\mathbf{y}_i - \bar{\mathbf{y}}_\nu)
  \right) \leq \text{ESS}(\mathbf{y}_i),
\end{equation}
therefore yields a smaller ESS estimate than the one corresponding to $\Sigma_{ff}$. We thus obtain a conservative diagnostic $\hat{\mathrm{ESS}}(\mathbf{y}_i)$, which can be interpreted as a lower bound on the true ESS and used as a safe diagnostic for guiding the grouping of outer samples.

We emphasize that evaluating the conservative diagnostic (Eq. \eqref{eq:hat-ess-bound}) itself is essentially free of forward model since the forecast covariance $P_{ff}$ is already produced by the EKI prediction step. More importantly, this diagnostic yields a substantial saving in the subsequent forward-model cost. In a naive strategy, one would first generate samples from the global GPP (forward step 1), propagate these samples through the forward solver (forward step 2), and estimate $\Sigma_{ff}$ from the resulting model outputs in order to diagnose the quality of importance sampling. If this diagnosis reveals severe weight degeneracy, the entire set of forward (step 2) evaluations for the global GPP proposal is wasted, since one must then construct grouped GPP and again propagate the grouped samples through the forward model. By contrast, our conservative diagnostic uses only forecast-level information (from step 1) to predict whether the global GPP is likely to be adequate. When $\hat{\mathrm{ESS}}(\mathbf{y}_i)$ indicates poor quality, we bypass the forward (step 2) evaluations of global GPP samples altogether and directly proceed to grouping. 

\textbf{Grouping outer samples}: To utilize the diagnostic \eqref{eq:hat-ess-bound}, we introduce a user-chosen threshold $S>0$ for the estimated ESS. Outer samples whose
\begin{equation}
    \hat{\mathrm{ESS}}(\mathbf{y}_i) < S
\end{equation}
are marked as \emph{problematic}, because they are expected to contribute disproportionately to the variance of the importance estimator under the current proposal. A practical choice for $S$ is on the order of parameter dimension $N^{d_{\boldsymbol\theta}}$. More generally, $S$ may also be selected based on empirical calibration or problem-dependent bounds motivated by effective dimension ideas in randomized linear algebra. 

If $\hat{\mathrm{ESS}}(\mathbf{y}_i) \ge S$ holds for the majority outer samples, we simply retain the global GPP as the proposal and proceed with importance sampling and gradient estimation without any grouping. Otherwise, we trigger the grouping step by collecting the problematic outer samples. Let
\begin{equation}
     \mathcal{I}_{\mathrm{prob}} \;=\; \bigl\{\, i : \hat{\mathrm{ESS}}(\mathbf{y}_i) < S \,\bigr\}
\end{equation}
denote the index set of problematic outer samples, and collect their associated observations into
\begin{equation}
    Y^{\mathrm{prob}} \;=\; \bigl\{\, \mathbf{y}_i : i \in \mathcal{I}_{\mathrm{prob}} \,\bigr\}.
\end{equation}

Instead of using a single pooled posterior for all problematic samples, we further partition $Y^{\mathrm{prob}}$ into $K \ge 2$ disjoint groups by clustering in observation space:
\begin{equation}
     Y^{\mathrm{prob}}
  \;=\;
  \bigsqcup_{k=1}^K Y_k,
  \qquad
  Y_k \;=\; \{\, \mathbf{y}_i : i \in \mathcal{I}_k \,\},
\end{equation}
where $\{\mathcal{I}_k\}_{k=1}^K$ forms a partition of $\mathcal{I}_{\mathrm{prob}}$ and the number of clusters $K$ is chosen such that each group collects outer samples with comparable observational behavior. For each group $Y_k$, we then construct a group-specific pooled posterior by pooling over the outer samples only in that group, while retaining the same pooling rule as in the global GPP construction.

\textbf{Generate samples for grouped GPP}. For each group of outer samples $Y_k$, $k=1,2,\dots,K$, we associate a group-specific GPP of the form
\begin{equation}
    p(\boldsymbol\theta | Y_k)\;\propto\;
  p(\boldsymbol\theta)\,\prod_{i \in \mathcal{I}_k} p(\mathbf{y}_i | \boldsymbol\theta)^{\nu'_i},\quad \Sigma_{i \in \mathcal{I}_k} \nu'_i=1
\end{equation}
where $\mathcal{I}_k$ denotes the index set of outer samples in group $k$, and $\nu'_i$ denotes the normalized weight. To generate samples from this grouped pooled posterior, we reuse the EKI update in Eq.~\eqref{eq:eki final}, but with the pooled observation mean replaced by $\bar{\mathbf{y}}_\nu^{(k)} \;=\; \sum_{i \in \mathcal{I}_k} \nu'_i\, \mathbf{y}_i $.

Applying Eq.~\eqref{eq:eki final} with $\bar{\mathbf{y}}_\nu^{(k)}$ and the same initial ensemble $\{\boldsymbol\theta^{(j)}\}_{j=1}^J$ yields an updated ensemble $\{\boldsymbol\theta^{(j)} | Y_k\}_{j=1}^J$, which we interpret as samples from the grouped pooled posterior corresponding to $Y_k$.

Importantly, this regeneration step does not require any additional forward-model simulations. This follows from the affine-map structure of EKI: as long as we reuse the same initial ensemble $\{\boldsymbol\theta^{(j)}\}_{j=1}^J$, the associated model forecasts $\{f(\boldsymbol\theta^{(j)})\}_{j=1}^J$ have already been computed in the original EKI run before grouping and can be reused for all groups. This property sharply distinguishes our ensemble-based construction from using conditional diffusion or MCMC as the sample generator when extending a single pooled posterior to multiple grouped posteriors. For example, in the conditional-diffusion setting, the forward-model cost of sample generation grows linearly with the number of groups, because each proposal sample follows its own backward SDE trajectory, each inner step requires evaluating the likelihood score and the forward model, and these forward evaluations cannot be reused across groups due to the independent noise realizations. In contrast, the proposed ensemble-based method is free of additional forward-model simulations when generating samples for extra groups, which is a key computational advantage of the grouped GPP–EKI framework.

\textbf{Gradient estimation}: Let $\mathcal{I}_{\mathrm{prob}} = \bigsqcup_{k=1}^K \mathcal{I}_k$ be the index set of problematic outer samples and let $\mathcal{I}_{\mathrm{ok}} = \{1,\dots,N\}\setminus \mathcal{I}_{\mathrm{prob}}$ denote the remaining outer samples for which the global GPP proposal is accepted. With the grouped GPP samples, the grouped estimator of the EIG gradient can be written as
\begin{equation}
  \label{eq:mc-eig-gradient-grouped}
  \begin{aligned}
  \nabla_{\mathbf{d}} I(\mathbf{d})
  \;\approx\; \frac{1}{N} \Bigg[
    &\sum_{i=1}^N \nabla_{\mathbf{d}} \log p(\mathbf{d},\mathbf{y}_i,\boldsymbol\theta_i) \\
    &- \sum_{i \in \mathcal{I}_{\mathrm{ok}}} \frac{1}{J} \sum_{j=1}^J
      \frac{p(\boldsymbol\theta'_j | \mathbf{y}_i,\mathbf{d})}
           {p_{\mathrm{GPP}}(\boldsymbol\theta'_j |Y_{\mathrm{ok}}, \mathbf{d})}\,
      \nabla_{\mathbf{d}} \log p(\mathbf{d},\mathbf{y}_i,\boldsymbol\theta_i,\boldsymbol\theta'_j) \\
    &- \sum_{k=1}^K \sum_{i \in \mathcal{I}_k} \frac{1}{J} \sum_{j=1}^J
      \frac{p(\boldsymbol\theta'_{k,j} | \mathbf{y}_i,\mathbf{d})}
           {p_{\mathrm{GPP},k}(\boldsymbol\theta'_{k,j} | Y_k, \mathbf{d})}\,
      \nabla_{\mathbf{d}} \log p(\mathbf{d},\mathbf{y}_i,\boldsymbol\theta_i,\boldsymbol\theta'_{k,j})
  \Bigg],
  \end{aligned}
\end{equation}
where $\{\boldsymbol\theta'_j\}_{j=1}^J$ are inner samples drawn from the grouped GPP $p_{\mathrm{GPP}}(\theta|\mathbf{d})$ based on $y_i, i \in \mathcal{I}_{\mathrm{ok}}$ and $\{\boldsymbol\theta'_{k,j}\}_{j=1}^J$ are inner samples drawn from the grouped pooled posterior $p_{\mathrm{GPP},k}(\theta|\mathbf{d})$ associated with $Y_k$.

\section{Sequential model discrepancy calibration framework}\label{sec: model discrepancy}

Model discrepancy is ubiquitous in real-world applications, posing challenges to both forward and inverse problems, not to mention active learning and BED. Under model discrepancy scenarios, we build on the framework in \cite{YANG2026114469} 

Specifically, to define the model discrepancy, the true system is written in general form
\begin{equation}
    \mathbf{y}=\mathcal{G}^\dagger(\boldsymbol{\theta}_\mathcal{G};\mathbf{d}),
\end{equation}
while the incorrect model is
\begin{equation}
    \mathbf{y}=\mathcal{G}(\boldsymbol{\theta}_\mathcal{G},\boldsymbol{\theta}_\mathcal{E} ;\mathbf{d}),
\end{equation}
where $\boldsymbol{\theta}_\mathcal{G}$ represents the original physical parameter, and $\boldsymbol{\theta}_\mathcal{E}$ represents the error parameter. Note that the system $\mathcal{G}^\dagger$ and model $\mathcal{G}$ may also serve as the right-hand-side of dynamic systems, and in that case design $\mathbf{d}$ usually appears in the additional measuring operator, rather than the input of dynamic systems. To calibrate the model discrepancy, we seek to optimize the error parameter $\boldsymbol{\theta}_\mathcal{E}$ such that the model $\mathcal{G}$ could approximate the system $\mathcal{G}^\dagger$.

In the context of BED, this framework determines optimal designs for physical parameters $\boldsymbol{\theta}_\mathcal{G}$ and error parameters $\boldsymbol{\theta}_\mathcal{E}$, respectively. The purpose of separating these two types of parameters rather than jointly seeking designs and updating is to avoid the high-dimensional joint parameter space $\{\boldsymbol{\theta}_\mathcal{G}, \boldsymbol{\theta}_\mathcal{E}\}$, especially when using a neural network to calibrate structural model discrepancy. At each stage in the sequential process, for relatively low-dimensional physical parameters, we retain standard BED methods to allow arbitrary distributions. For high-dimensional error parameters, we apply the GPP-EKI to find the optimal designs and use a gradient-based method to update values of error parameters. These two steps are repeated in each stage. The updated model benefits the updating in later stages, and the updated belief allows more effective model correction.

At each stage of sequential BED, the optimal design for physical parameters is determined using standard BED: 
\begin{equation}
\label{eq:BED_optimization_G}
    \mathbf{d}^*_\mathcal{G} = \argmax_{\mathbf{d}\in\mathcal{D}} \mathbb{E}[U(p(\boldsymbol{\theta}_\mathcal{G}|\mathbf{y};\boldsymbol{\theta}_\mathcal{E}) || p(\boldsymbol{\theta}_\mathcal{G})  )],
\end{equation}
where the posterior distribution of $\boldsymbol{\theta}_\mathcal{G}$ from the previous stage serves as the prior distribution $p(\boldsymbol{\theta}_\mathcal{G})$ for the current stage. The utility function $U$ could either be the KL divergence or the Wasserstein distance. Considering that $\boldsymbol{\theta}_\mathcal{G}$ is often low-dimensional, standard BED methods can be employed to solve the optimization problem in Eq.~\eqref{eq:BED_optimization_G}. With the optimal design $\mathbf{d}^*_\mathcal{G}$, the corresponding data $\mathbf{y}_\mathcal{G}$ is then used to update the belief of physical parameters by the Bayesian theorem conditioned on the current network coefficients $\boldsymbol\theta_\text{NN}$:
\begin{equation}
    p(\boldsymbol\theta_\mathcal{G}|\mathbf{y}_\mathcal{G},\mathbf{d}^*_\mathcal{G};\boldsymbol\theta_\mathcal{E})=\frac{p(\mathbf{y}_\mathcal{G}|\boldsymbol\theta_\mathcal{G},\mathbf{d}^*_\mathcal{G};\boldsymbol\theta_\mathcal{E})p(\boldsymbol\theta_\mathcal{G})  }{p(\mathbf{y}_\mathcal{G}|\mathbf{d}^*_\mathcal{G};\boldsymbol\theta_\mathcal{E})},
\end{equation}
which also provides the MAP estimation of physical parameters $\boldsymbol\theta_\mathcal{G}^*$:
\begin{equation}
    \label{eq:selet 1 theta}\boldsymbol\theta_\mathcal{G}^*=\argmax_{\boldsymbol\theta_\mathcal{G}}
 \{ p(\boldsymbol\theta_\mathcal{G}|\mathbf{y}_\mathcal{G},\mathbf{d}^*_\mathcal{G};\boldsymbol\theta_\mathcal{E})\}.
\end{equation}

On the other hand, the optimal design for neural network parameters is found by BED with the GPP-EKI:
\begin{equation}
    \begin{aligned}
        \mathbf{d}^*_\mathcal{E} &= \argmax_{\mathbf{d}\in\mathcal{D}} \mathbb{E}[U(p(\boldsymbol{\theta}_\mathcal{E}|\mathbf{y};\boldsymbol{\theta}_\mathcal{G}^*) || p(\boldsymbol{\theta}_\mathcal{E})  )],\\
    \end{aligned}
\end{equation}

At each stage of the sequential BED, we assume a Gaussian prior over the error parameters centered at their current value as their prior distribution. Once the design $\mathbf{d}^*_\mathcal{E}$ is determined, the resulting data $\mathbf{y}_\mathcal{E}$ is then used to update the network coefficients via standard gradient-based optimization:
\begin{equation}
    \boldsymbol\theta^*_\mathcal{E}=\argmax_{\boldsymbol\theta_\mathcal{E}} p(\mathbf{y}_\mathcal{E}|\boldsymbol\theta_\mathcal{G}^*,\mathbf{d}^*_\mathcal{E},\boldsymbol\theta_\mathcal{E}).
\end{equation}

At each stage, we immediately use the updated model ($\boldsymbol\theta^*_\mathcal{E}$) and the accumulated data for physical parameters, $\mathbf{y}_\mathcal{G}$, to refine the belief of the physical parameter $\boldsymbol{\theta}_\mathcal{G}$. Specifically, the physical parameter posterior is recomputed as:
\begin{equation}
p(\boldsymbol\theta_\mathcal{G}|\mathbf{Y}_\mathcal{G},\mathbf{D}_\mathcal{G};\boldsymbol\theta^*_\mathcal{E})=\frac{p(\mathbf{Y}_\mathcal{G}|\boldsymbol\theta_\mathcal{G},\mathbf{D}_\mathcal{G};\boldsymbol\theta^*_\mathcal{E})p(\boldsymbol\theta_\mathcal{G})  }{p(\mathbf{Y}_\mathcal{G}|\mathbf{D}_\mathcal{G};\boldsymbol\theta^*_\mathcal{E})},
\end{equation}
where $\mathbf{D}_\mathcal{G}=\{\mathbf{d}_\mathcal{G}^n,\mathbf{d}_\mathcal{G}^{n-1},\cdots,\mathbf{d}_\mathcal{G}^1\}$ and $\mathbf{Y}_\mathcal{G}=\{\mathbf{y}_\mathcal{G}^n,\mathbf{y}_\mathcal{G}^{n-1},\cdots,\mathbf{y}_\mathcal{G}^1\}$ denote the sequence of historical designs and corresponding observations, and $p(\boldsymbol\theta_\mathcal{G})$ is the original prior, which may be uniform.

\section{Experimental results}
\label{sec: Numerical Results}

We study a classical source inverse problem in the BED community~\cite{huan_gradient-based_2014, shen_bayesian_2023}, with specially designed setups of model discrepancy. A contaminant source is put into a time-varying 2-dimensional convection-diffusion field. The source inverse problem is to determine the optimal measurement location for a concentration value and infer the source location, possibly other information about this system as well~\cite{shen_bayesian_2023, yang_active_2025}. Model discrepancy could appear in different forms~\cite{YANG2026114469}, i.e., incorrect value of known parameters, or incorrect function form. The former one is named as a parametric error, where the model function form is still assumed to be correct. And the latter one is named as a structural error, which requires more advanced techniques to address. In this work, we assign the source location as the physical parameters and tackle the data choice by standard BED methods, and apply the proposed method to determine optimal design only for the error parameters.

Specifically, the system, governing equation of the concentration field $\mathbf{u}$, is written as:
\begin{equation}
    \frac{\partial \mathbf{u}(\mathbf{z},t;\boldsymbol\theta)}{\partial t}=D\nabla^2\mathbf{u}-\mathbf{v}(t) \cdot \nabla \mathbf{u}+S(\mathbf{z},t;\boldsymbol\theta),~~~\mathbf{z} \in [z_L,z_R]^2,~~t>0
    \label{eq:true_system_example},
\end{equation}
where $D$ is the diffusion coefficient assumed to be known with a true value of $1$, $\mathbf{v}=\{v_x,v_y\} \subseteq \mathbb{R}^2$ is a time-dependent convection velocity assumed to be known with true value of $v_x=v_y=50t$, $S$ denotes the source term with some parameters $\boldsymbol\theta$. In this work, the true system has an exponentially decay source term in the following form with the parameters $\boldsymbol\theta=\{\theta_x,\theta_y,\theta_h,\theta_s\} \in \mathbb{R}^4$:
\begin{equation}
    S(\mathbf{z},t;\boldsymbol\theta)=\frac{\theta_s}{2\pi\theta_h^2}\exp \left(-\frac{(\theta_x-z_x)^2+(\theta_y-z_y)^2}{2\theta_h^2} \right),
    \label{eq:true_source_term}
\end{equation}
where \(\theta_x\) and \(\theta_y\) denote the source location, \(\theta_h\) and \(\theta_s\) represent the source width and source strength. The solution field is defined on $[z_L,z_R]^2=[-3,2]^2$. The initial condition is \(\mathbf{u}(\mathbf{z}, 0;\boldsymbol\theta) = \mathbf{0}\), and a homogeneous Neumann boundary condition is imposed for all sides of the square domain. 

The measuring process is written as:
\begin{equation}
    y=\mathbf{u}(\mathbf{d},t)+\epsilon, \quad \epsilon \sim \mathcal {N}(0,\sigma^2),
\end{equation}
where $\mathbf{d}$ denotes the design variable, representing the measure location, which is a 2-dimensional spatial coordinate of the solution field. The temporal coordinates $t$ of the measurement are set as $0.05+n\times 0.005$ time units, where $n$ is the stage number ($n=1,2,\cdots$). At each stage, each experiment/design corresponds to a single measurement. Measurements contain a Gaussian noise $\epsilon$ with fixed variance $\sigma=0.05^2$.

For the physical parameters $\{\theta_x, \theta_y\}$, the inverse problem remains nonlinear with non-Gaussian posteriors in all experiments. The error parameters, however, follow different setups in the two test cases. The key motivations and conclusions of the numerical results are summarized below:

\textbullet\ We examine the use of EKI to generate samples from the pooled posterior and the use of this pooled posterior as an importance-sampling proposal in place of the prior. In a low-dimensional Gaussian–linear error parameter case with analytical expressions available, we show that the EKI ensemble faithfully approximates the pooled posterior, and that this pooled posterior is systematically closer to the individual posteriors than the prior (or at worst comparable). Detailed results can be found in Section~\ref{sec:Correct parametric error}.

\textbullet\ We investigate the effect of the proposed grouping strategy on the mismatch between proposal and target distributions and on the variance of the importance-sampling–based gradient estimator. In a structural error (incorrect model form) case with a high-dimensional neural-network correction, we show that grouping reduces proposal–target discrepancy, improves effective sample usage, and substantially lowers the variance of the estimated design gradients at a given computational cost. Detailed results can be found in Section~\ref{sec:Correct functional error}.

\textbullet\ We validate the proposed Bayesian experimental design framework by applying it to both parametric and structural model-error scenarios and comparing it with AD-EKI and diffusion-based baselines. Across these examples, our method achieves posteriors, corrected solution fields, and design choices that are comparable to those of the baselines, while requiring substantially fewer PDE solves, thereby demonstrating that the proposed components combine into an efficient and reliable BED scheme. Detailed results can be found in Sections~\ref{sec:Correct parametric error} and~\ref{sec:Correct functional error}.

\subsection{Parametric case}\label{sec:Correct parametric error}
This section considers an example with the discrepancy represented by a low-dimensional Gaussian–linear error parameter, while the physical parameters enter the forward model nonlinearly and thus have non-Gaussian posteriors. We use this setting to validate the proposed method without grouping the outer samples and compare with other baseline methods, i.e., AD-EKI~\cite{YANG2026114469} and contrastive diffusion~\cite{iollo2024bayesian}. In this example, we first examine the consistency of the resulting posteriors and designs across methods, then assess the quality of EKI sampling of the pooled posterior, and finally compare the PDE cost and analyze the limitations of using a single global proposal.

For the model discrepancy setup, we assume a parametric error in the source strength $\theta_s$: the forward model uses an initial value of $3$, whereas the true value in the system is $2$. The task is to infer the source location $\{\theta_x, \theta_y\}$ and correct the error parameter $\theta_s$. Note that the mapping from the error parameter to the measurement, $\theta_s \mapsto \mathbf{y}$, is linear, so this case serves as a toy example in which analytical expressions for many quantities of interest are accessible for comparison. This setup could also be cast as a standard BED problem with the joint unknown $\{\theta_x, \theta_y, \theta_s\}$~\cite{yang_active_2025}, at the cost of a higher-dimensional parameter space. We use the present parametric-error formulation as a simple starting point to verify the proposed method.

For the sampling configuration, we use 180 outer samples and an inner ensemble of size 180, which is sufficient for the one-dimensional error parameter in this example. This setup is shared by our GPP-EKI (without grouping) and the contrastive diffusion method. For AD-EKI, we use 30 outer samples, an inner ensemble size of 30, and 3 EKI iterations, which we found to be a relatively small configuration that is still sufficient in this case.

Overall, the proposed method achieves results comparable to those of the contrastive diffusion and AD-EKI baselines, but with substantially lower computational cost. We first present the evolution of the posterior distributions of the physical parameters (Fig.~\ref{fig: posterior parametric}) and the values of the error parameters (Fig.~\ref{fig: error parametric}) as an overall validation. For all three methods, the posteriors over the physical parameters and the inferred error parameters converge to the true values, demonstrating that the inverse problem under model discrepancy has been successfully solved. Moreover, the posterior patterns and design choices are qualitatively similar across methods, indicating that the estimated design gradients are close to each other. These observations support the effectiveness of the proposed gradient-estimation strategy, in which EKI-generated samples from the pooled posterior are used as proposals for importance sampling.

\begin{figure}[H]
    \centering
    \includegraphics[width=\linewidth]{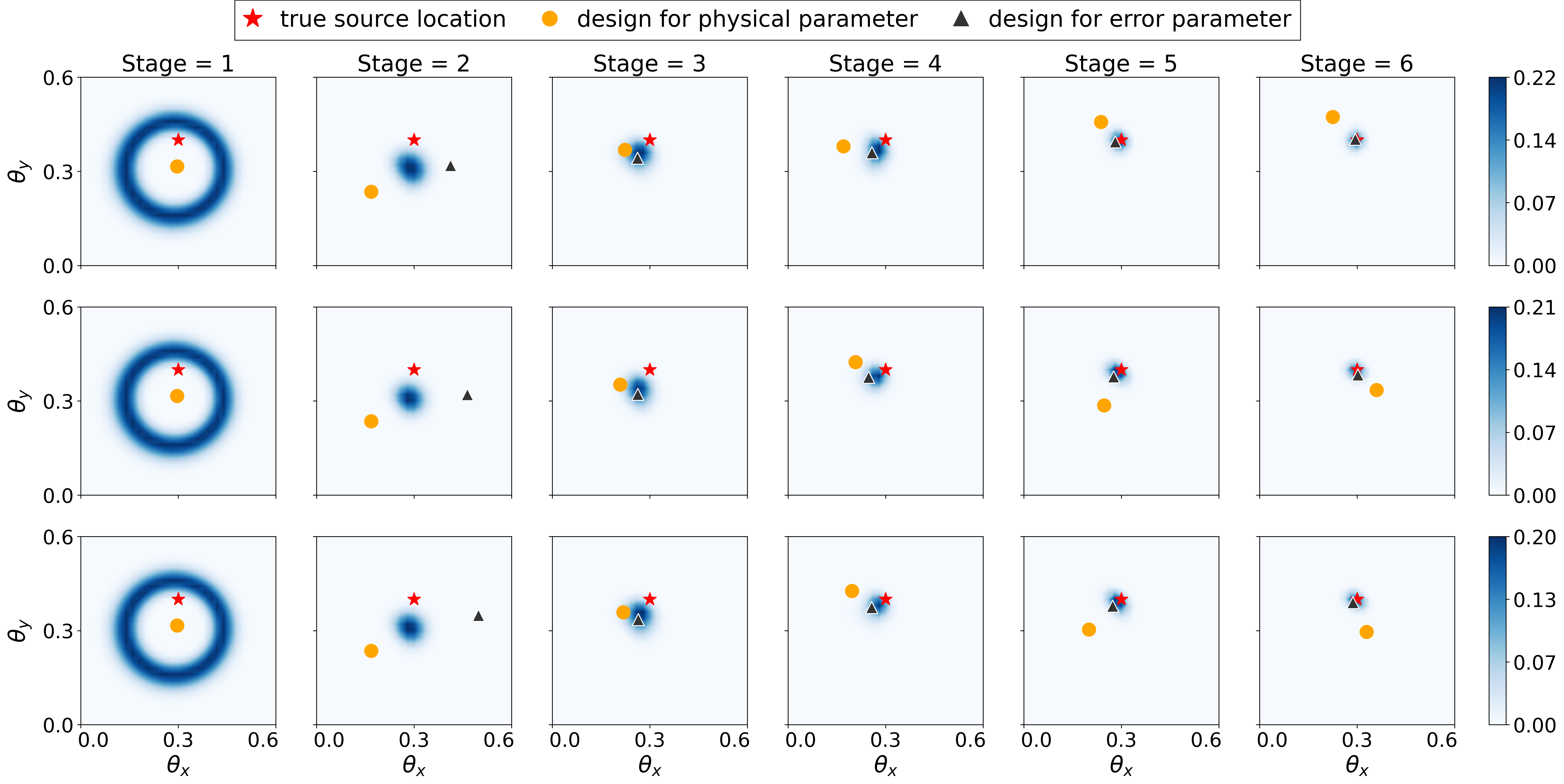}
    \caption{Posteriors of the physical parameter. From top to bottom, the rows present the estimates obtained by contrastive diffusion, AD-EKI, and GPP-EKI. The beliefs are defined on $[0,1]^2$ and shown zoomed in $[0, 0.6]^2$ for clarity.}
    \label{fig: posterior parametric}
\end{figure}

\begin{figure}[H]
    \centering
    \includegraphics[width=0.5\linewidth]{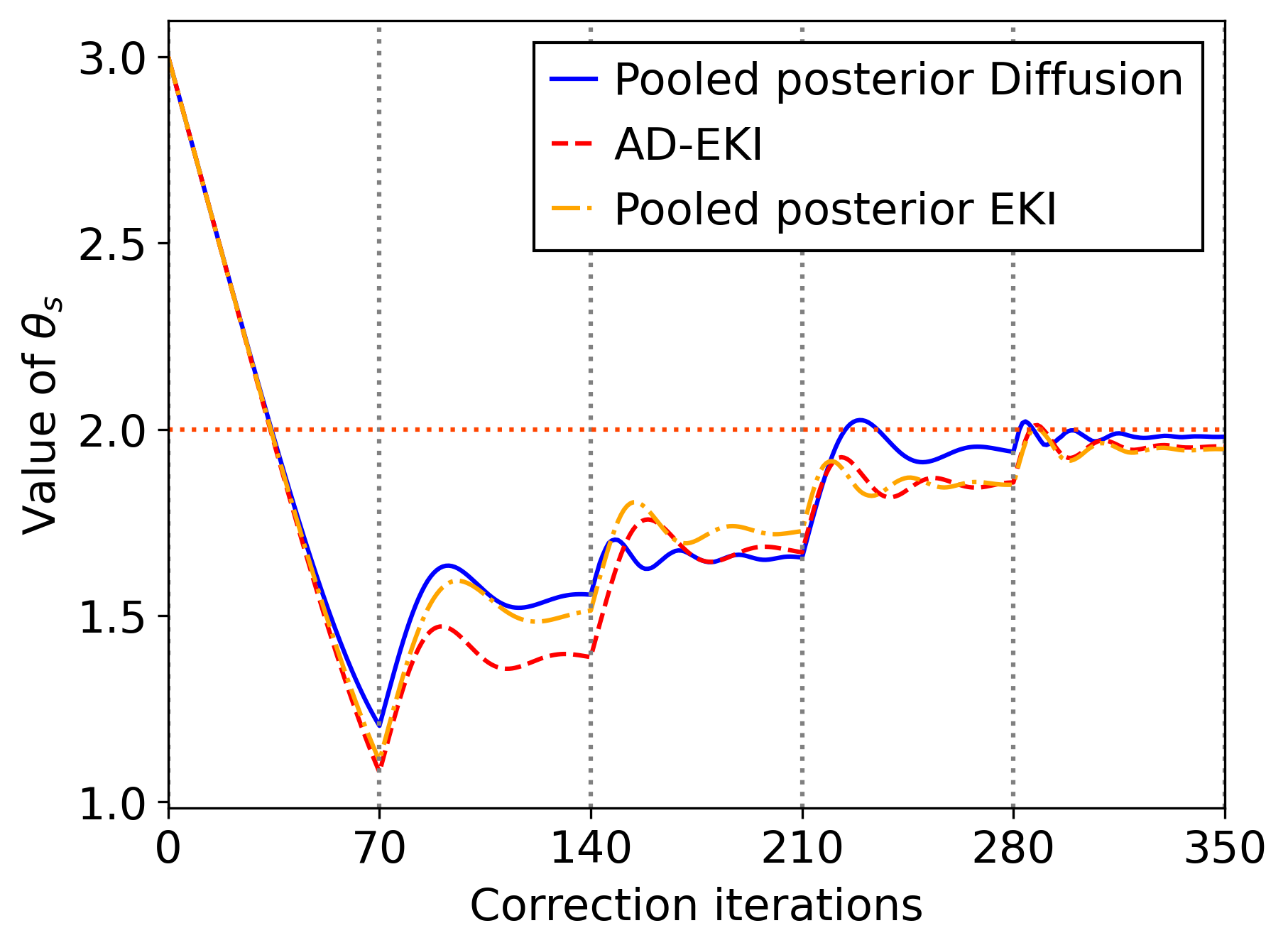}
    \caption{Updating trajectory of the error parameter.}
    \label{fig: error parametric}
\end{figure}

We then examine in more detail the pooled posterior samples generated by EKI. Figure~\ref{fig:samples_eki_pp} shows that the empirical distribution of the EKI ensemble closely matches the analytical pooled posterior along the error parameter $\theta_s$. The analytical pooled posterior has mean $3.00$ and variance $0.46$, whereas the EKI ensemble with $180$ particles yields mean $3.094$ and variance $0.424$, which is consistent with Monte Carlo sampling error at this ensemble size. This agreement indicates that the proposed sampling procedure is reliable in this Gaussian–linear setting.

\begin{figure}[H]
    \centering
    \begin{minipage}{0.48\linewidth}
        \centering
        \includegraphics[width=\linewidth]{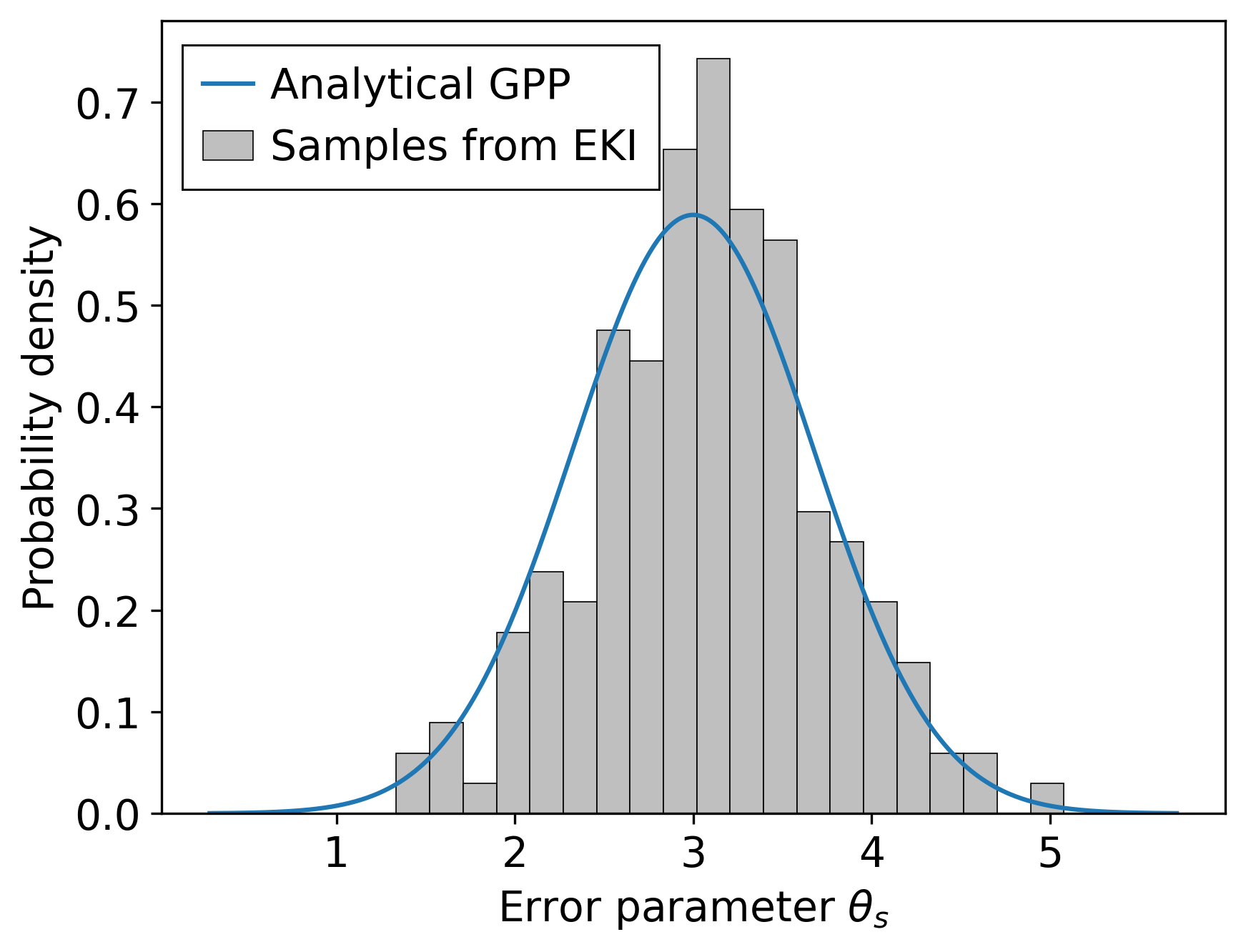}
        \caption{Samples for geometric pooled posterior drawn by EKI}
        \label{fig:samples_eki_pp}
    \end{minipage}
    \hfill
    \begin{minipage}{0.4\linewidth}
        \centering
        \includegraphics[width=\linewidth]{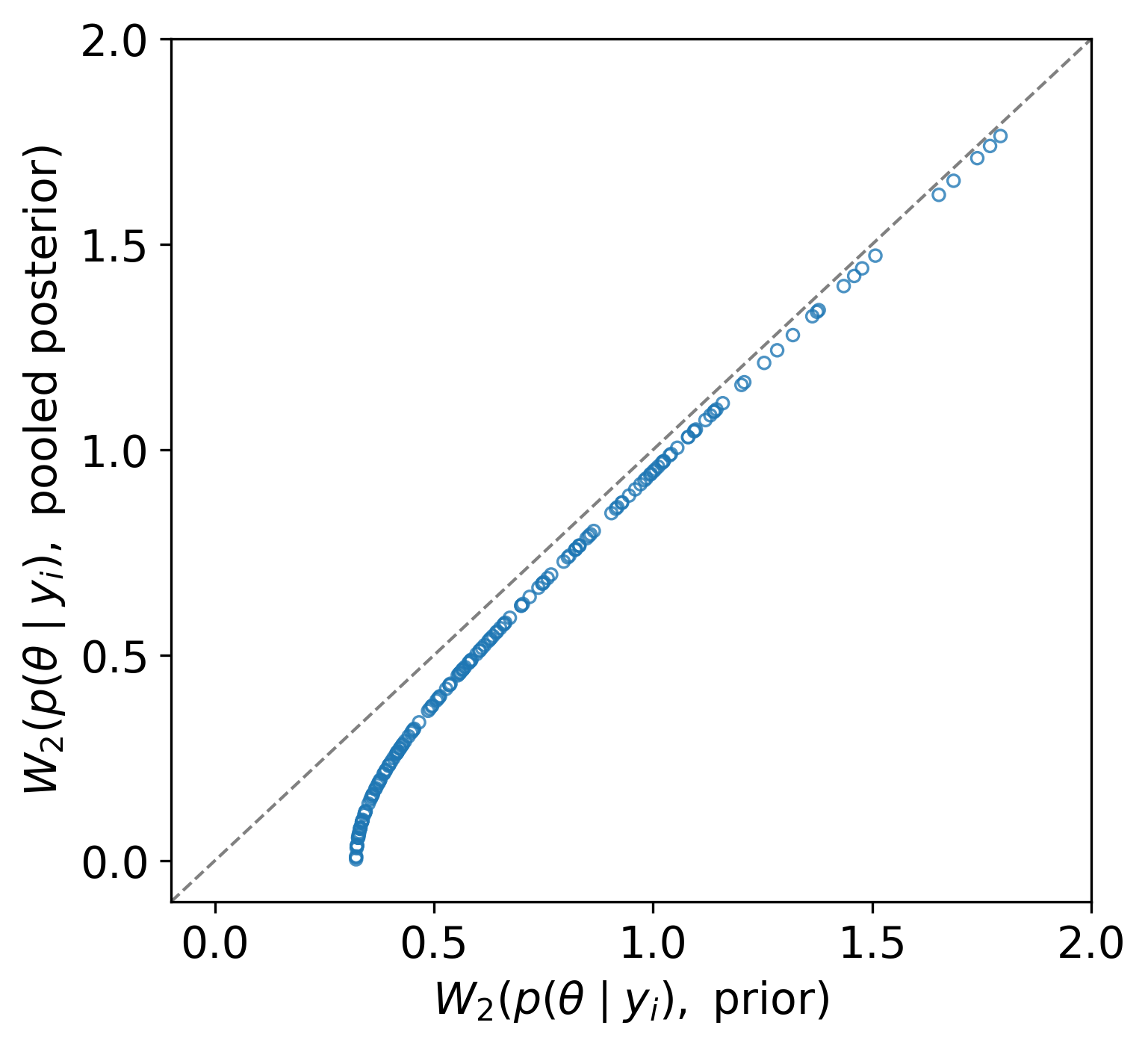}
        \caption{Scatter plots comparing, for each outer sample, the distance from its individual posterior to the global geometric pooled posterior versus the distance to the prior.}
        \label{fig:distance_pp_prior}
    \end{minipage}
\end{figure}

The PDE costs of the different methods, reported in Table~\ref{tab:pde_cost_parametric_error}, highlight the computational advantage of the proposed approach. For the diffusion-modeling method, a large number of PDE simulations is required because each sample at each backward SDE step needs a forward PDE solve to evaluate the likelihood score. For AD-EKI, the cost is substantially reduced: the small ensemble size and early stopping lead to fewer PDE solves when approximating the posterior for each outer sample. However, the total cost remains high because the PDE simulations for different outer samples cannot be reused. In contrast, the proposed GPP method is an importance sampling scheme that allows the inner samples from the proposal distribution to be reused across outer samples. This reuse substantially reduces the total number of PDE solves. The trade-off is that a single set of inner samples from a fixed proposal distribution may not be equally well matched to all individual posteriors, especially when these posteriors differ substantially across outer samples.

\begin{table}[H]
  \centering
  \caption{PDE cost in the parametric error case}
  \label{tab:pde_cost_parametric_error}
  \begin{tabular}{lccc}
    \toprule
     & Diffusion & AD-EKI & GPP-EKI \\ \midrule
    cost & 36360 & 2700 & 540 \\ \bottomrule
  \end{tabular}
\end{table}

Adopting the pooled posterior as the proposal is one strategy to mitigate this proposal–target mismatch and is theoretically preferable to using the prior as an empirical proposal. Figure~\ref{fig:distance_pp_prior} reports, for each outer sample $y$, the Wasserstein distance between its individual posterior and either the prior (horizontal axis) or the pooled posterior (vertical axis). All points lie below the diagonal, indicating that the pooled posterior is systematically closer to the individual posteriors than the prior and thus provides a better global proposal. However, for tail observations with large prior–posterior distance, the discrepancy remains substantial even when using the pooled posterior. In this low-dimensional setting with a moderately large ensemble, the resulting importance weights are still well behaved and the design gradient is not significantly affected, but the example already reveals a structural limitation of using a single global proposal. This motivates the grouped GPP variant introduced in the next section, where we further investigate its behavior in a higher-dimensional example.

\subsection{Structural case}\label{sec:Correct functional error}

In this section, we consider a structural error case where the model discrepancy is represented by a high-dimensional neural-network correction term, leading to strongly non-Gaussian posteriors. Such a high-dimensional parameterization, combined with a limited outer-sample budget, further amplifies weight degeneracy, making it difficult for a single global proposal to adequately match the complex target distribution. To address this issue, we group the outer samples to construct localized proposals that more systematically improve the alignment between proposal and target. We use this example to assess the performance of the proposed grouped GPP-EKI scheme against modified AD-EKI and random design; the results below first show the overall behavior of the posteriors and corrected fields, and then summarize the effects of grouping on proposal quality, estimator variance, and PDE cost.

For model discrepancy, we assume the known forcing term as:
\begin{equation}
    S(\mathbf{z},t;\bm\theta)=\frac{3\theta_s}{\pi\left(\frac{(\theta_x-z_x)^2+(\theta_y-z_y)^2}{2\theta_h^2}+2\theta_h^2\right)},
    \label{eq: c sourse}
\end{equation}
which differs from the exponentially decaying source term as defined in Eq.~\eqref{eq:true_source_term} for the true system. Without addressing this model structural error, the inference of source location through standard BED yields biased results. In this work, a neural network $\mathbf{NN}(z_x,z_y,\theta_x,\theta_y;\boldsymbol{\theta_\text{NN}})$ is employed to characterize the model structural error. We utilize a fully connected neural network consisting of 37 parameters, a configuration that introduces a higher-dimensional parameter space than the example of parametric model error in Section~\ref{sec:Correct parametric error}. The neural network takes a 2-dimensional input vector, representing the relative distance from a specific grid point to the source location. The output is a scalar value for this point, indicating the discrepancy between the true and modeled source term values. The fully connected network is designed with two hidden layers, each containing four neurons, and employs a tanh activation function for both layers.

For the sampling configuration of the GPP-EKI method, we employ approximately $5000$ outer samples and an inner ensemble of size $5000$, and partition the outer samples into three groups for illustration. For the modified AD-EKI, we draw $5000$ outer parameter samples and map them to $200$ weighted grid points in the data space, with an inner ensemble size of $60$ and $2$ EKI iterations. The same learning rate and gradient-ascent step size on the design are employed for all methods to ensure a fair comparison. The choice of $5000$ samples is modest relative to the $37$-dimensional neural-network parameter space, so this experiment is deliberately conducted in a data-limited regime where the benefits of grouping are most relevant. Further implementation details are provided in our code.

We first examine the evolution of the beliefs on the physical parameters under different design strategies, as shown in Fig.~\ref{fig: posterior structural}. We compare three approaches: random design of model correction experiments, modified AD-EKI, and the proposed GPP-EKI. For both AD-EKI and GPP-EKI, the selected designs and the resulting posteriors over the physical parameters are almost indistinguishable and gradually concentrate around the true parameter values. This indicates that, in this structural error setting, GPP-EKI achieves design choices and gradient estimates that are comparable to those produced by AD-EKI. In contrast, when the designs for model correction are chosen purely at random, the posterior fails to concentrate around the true parameter, which further highlights the importance of appropriate design selection in the presence of model discrepancy.

\begin{figure}[H]
    \centering
    \includegraphics[width=\linewidth]{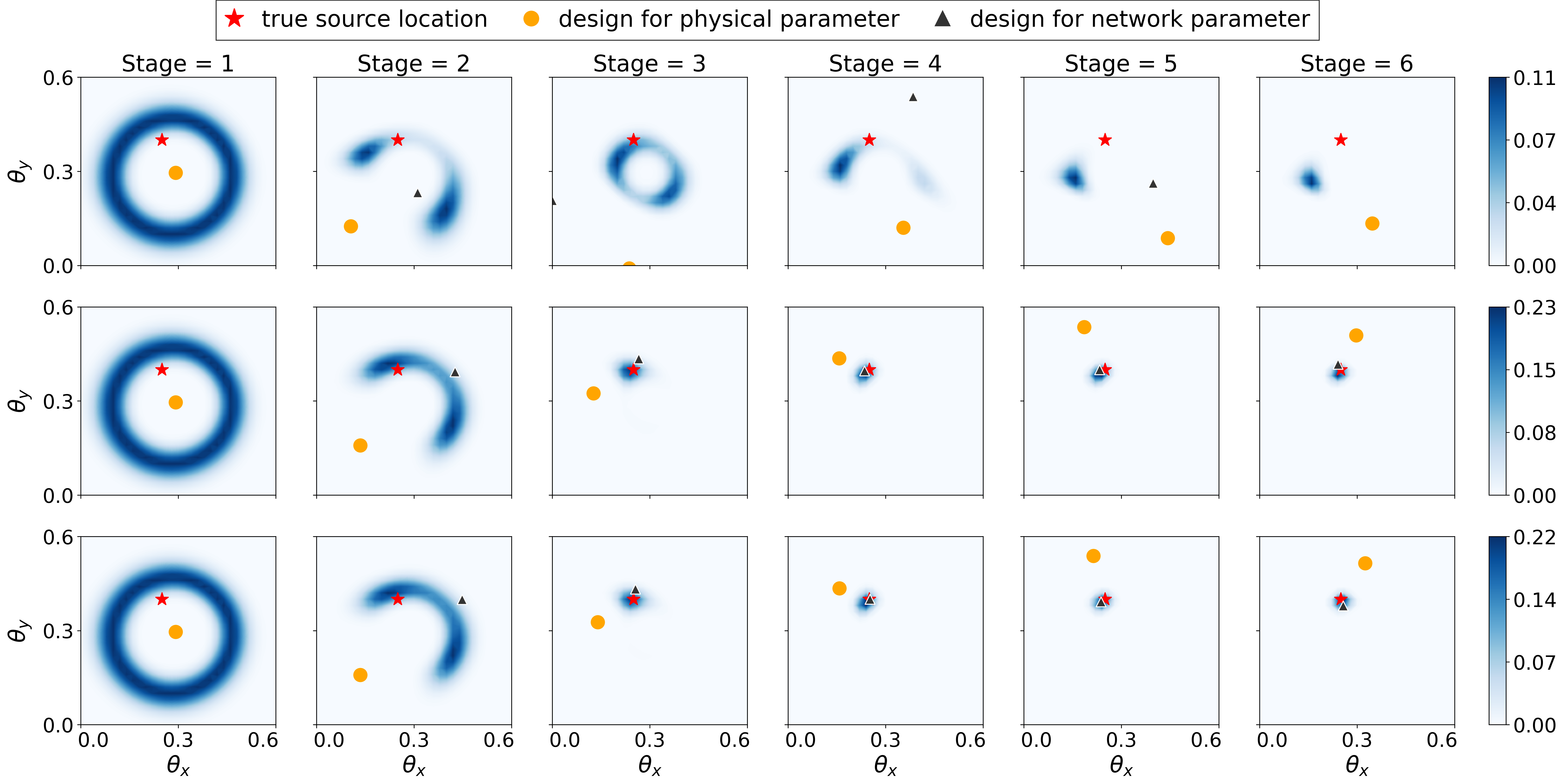}
    \caption{Posteriors of the physical parameter. From top to bottom, the rows present the estimates obtained by random selection, modified AD-EKI, and GPP-EKI. The beliefs are defined on $[0,1]^2$ and shown zoomed in $[0, 0.6]^2$ for clarity.}
    \label{fig: posterior structural}
\end{figure}

We then inspect the relative error between the corrected solution field and the true field, see Fig.~\ref{fig: structural_field_error}. Since the neural-network correction term has no “true” parameter for direct comparison, we focus on the resulting PDE solutions. Overall, the designs selected by AD-EKI and GPP-EKI yield solution fields whose relative errors are significantly smaller than those obtained with random design, especially in the vicinity of the design locations used to update the beliefs on the physical parameters. This shows that the model correction is already effective with minimal data. It also supports our previous conclusion that, even while reusing inner samples across outer samples, GPP-EKI still provides gradient estimates of comparable quality to those obtained by generating fresh inner ensembles for each outer sample as in AD-EKI, while keeping the computational cost low.

\begin{figure}[H]
    \centering
    \includegraphics[width=\linewidth]{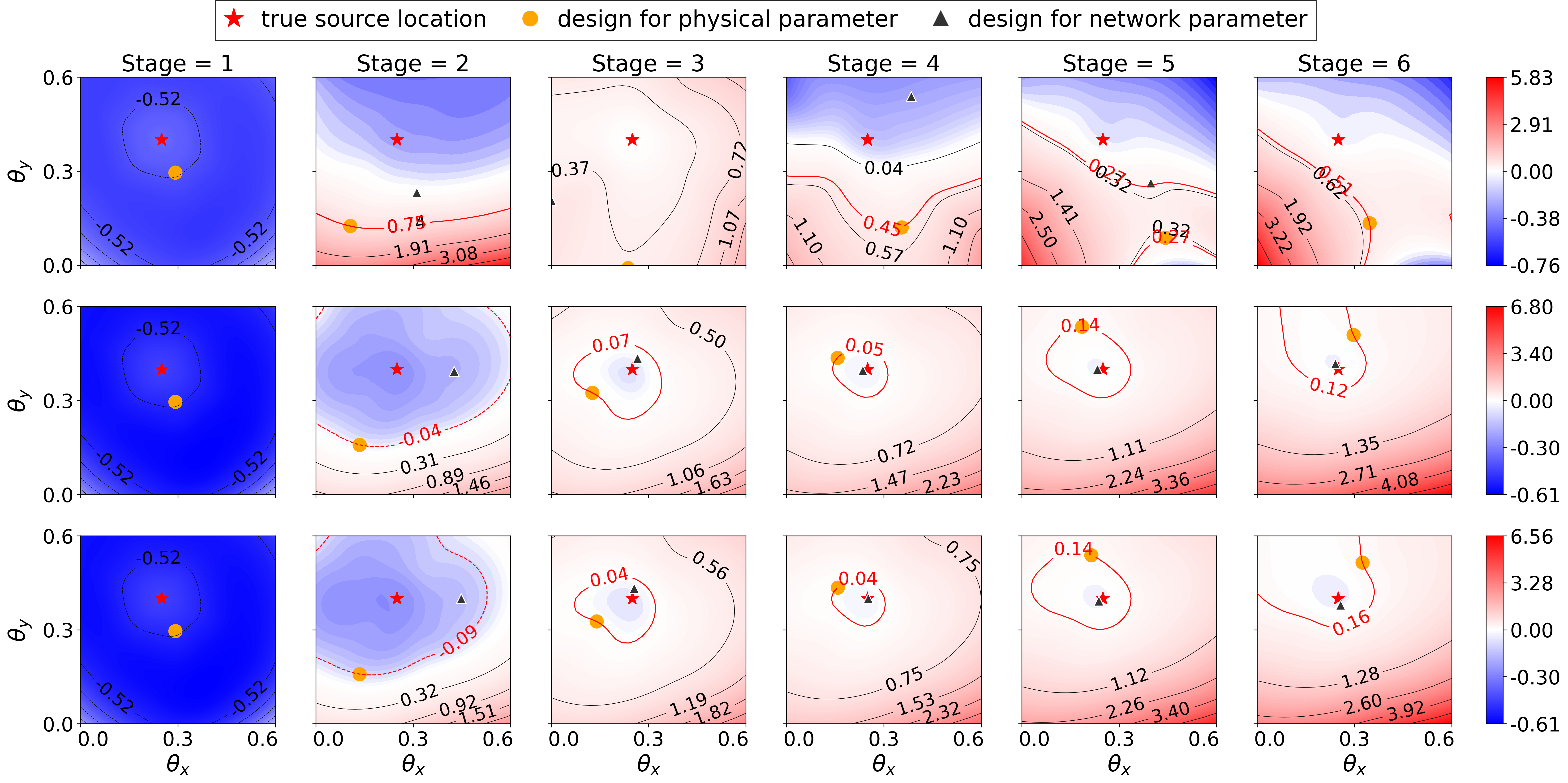}
    \caption{Relative error between the solution field and the true field. From top to bottom, the rows present the estimates obtained by contrastive diffusion, AD-EKI, and GPP-EKI. The beliefs are defined on $[0,1]^2$ and shown zoomed in $[0, 0.6]^2$ for clarity. The contour value at the design point in first row stage 3 is 0.20.}
    \label{fig: structural_field_error}
\end{figure}

Table~\ref{tab:pde_cost_structural_error} compares the PDE costs of the different methods. For the diffusion modeling approach used to sample from the pooled posterior, we adopt a basic reverse SDE sampler with 200 time steps. To obtain 5000 inner samples, this requires approximately 200×5000 forward-model evaluations. For a realistic nonlinear PDE model, such a cost is essentially prohibitive. More advanced sampling method of diffusion modeling may help reduce this cost. For AD-EKI, the cost is also high because a large number of outer samples must be treated separately. Although each outer sample requires only a relatively small ensemble (60) and can often be handled with early stopping, the total cost still scales multiplicatively with the number of outer samples (5000×60×2), and hence remains comparable to that of diffusion modeling. In return, AD-EKI provides high-quality inner gradient estimates for each individual outer sample. To overcome this constraint, we also consider a modified AD-EKI that exploits the low dimensionality of the observation space. Instead of drawing many random outer samples in parameter space, we construct a quadrature-like grid in the observation space and assign appropriate weights to these nodes, which substantially reduces the effective number of outer samples. This modification brings the PDE cost down to a more acceptable level, but it relies crucially on the observation space being low-dimensional and therefore does not extend to more general cases.

In contrast, the proposed GPP-EKI method does not rely on the dimension of the observation space. Here is a detailed breakdown. With 5000 outer samples, each requires one PDE solve to generate its synthetic observation. To construct three grouped pooled posteriors with 5000 inner samples per group, we first propagate an ensemble from prior through the forward model once, use the intermediate EKI quantities to estimate the quality of the importance sampling weights, and then perform grouping based on these diagnostics. The grouping step itself does not require any additional PDE solves. After grouping, we obtain three pooled posteriors with a total of 3×5000=15000 inner samples, using only 5000 PDE solves in this stage. Evaluating the design gradient then requires forward solves for all inner samples at the chosen design locations, which contributes another 15000 PDE solves. This evaluation step is required for any importance sampling scheme that uses a non-prior proposal, which is not specific to GPP-EKI. In total, the PDE cost of GPP-EKI in this example is therefore 5000+5000+15000=25000 solves, which is of the same order as that of the modified AD-EKI but achieved without relying on any special trick in observation space. Overall, GPP-EKI requires far fewer PDE solves for sample generation than the contrastive diffusion method by the feature that forming grouped pooled posteriors does not incur additional PDE cost, but also achieves a similar performance as AD-EKI.

\begin{table}[H]
  \centering
  \caption{PDE cost in structural error case}
  \label{tab:pde_cost_structural_error}
  \begin{tabular}{lccccc}
    \toprule
     & Random & Diffusion & AD-EKI & (Modified) AD-EKI & GPP-EKI \\ \midrule
    cost & - & 1,000,000 & 600,000 & 30,000 & 25,000\\ \bottomrule
  \end{tabular}
\end{table}

Figure~\ref{fig: lower bound} illustrates the grouping of outer samples $y$ induced by the proposed diagnostic criterion. We use the conservative diagnostic to predict the importance sampling performance and select groups directly in the observation space. The resulting grouping closely approximates the grouping that would be obtained by actually running importance sampling and then clustering based on the realized weights, but without the need to perform this expensive preliminary importance sampling step. Moreover, the groups clearly separate the main mass of the $y$ distribution from its tail structures, so that each group has a pooled posterior that is significantly better matched to its members than the global pooled posterior. In practice, one may introduce a relaxation parameter or other selection heuristics to obtain more balanced group sizes, although we do not explore such refinements here.

\begin{figure}[H]
    \centering
    \includegraphics[width=0.55\linewidth]{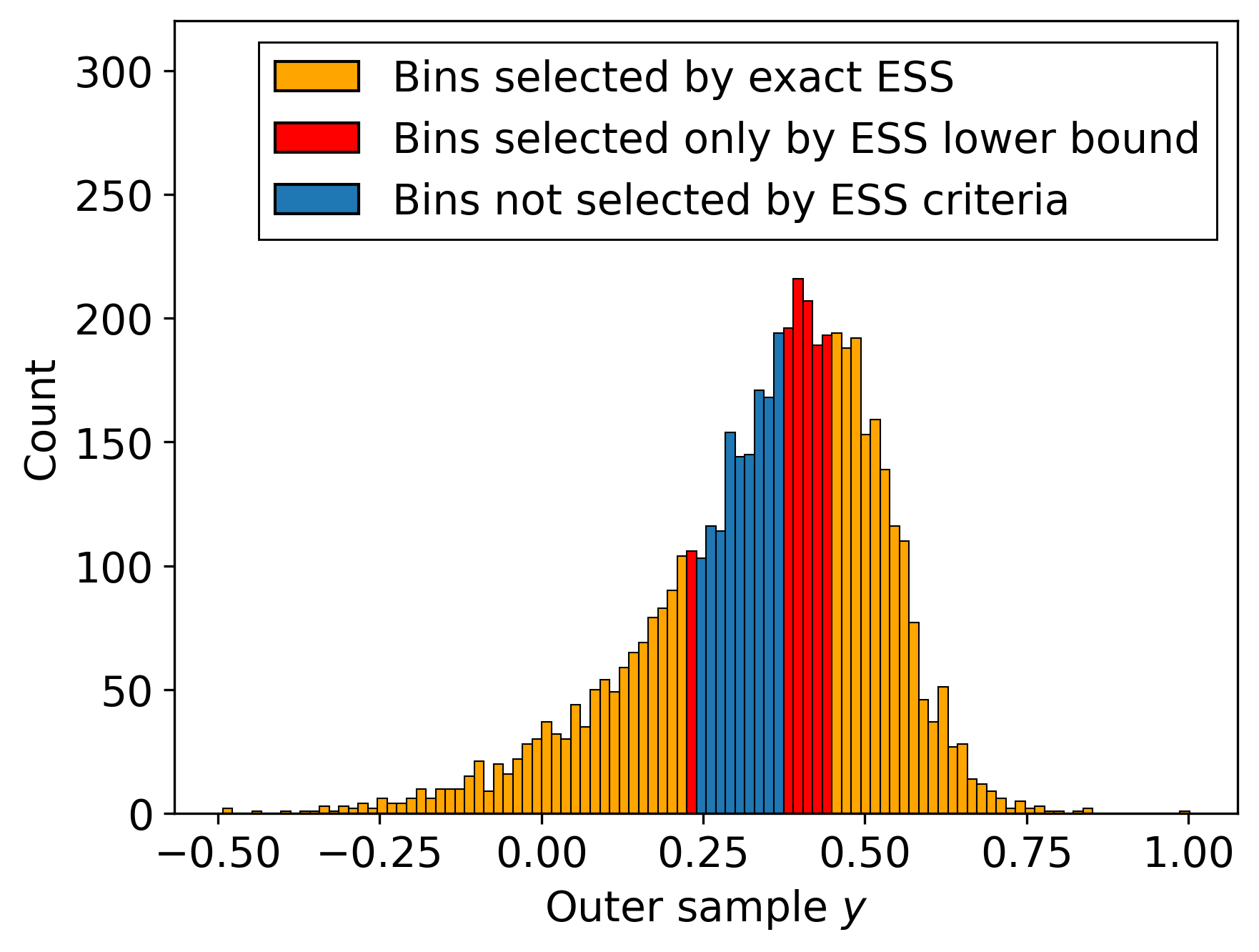}
    \caption{Histogram of all outer‐sample values $y$, annotated by their selection status under the proposed ESS diagnostics. Left orange+red region is Group 1, right orange+red region is Group 2, and blue bins form Group 3.}
    \label{fig: lower bound}
\end{figure}

To quantify the benefit of grouping, we examine, for every outer sample in each group, the distance between its individual posterior and both the global pooled posterior and the corresponding grouped pooled posterior, see Fig.~\ref{fig: no error reward map}. For the outer samples in the tail regions of the $y$ distribution, that is, in groups 1 and 2, the distance to the grouped pooled posterior is markedly smaller than the distance to the global pooled posterior. Equivalently, the grouped pooled posterior overlaps much more strongly with the individual posteriors than the global pooled posterior does, which leads to a substantial improvement in the quality of the importance sampling estimates. There remain a few outer samples whose individual posteriors are still relatively far from even the grouped pooled posterior. This residual mismatch is a common limitation of methods that reuse inner samples across outer samples. As a refinement, such problematic outer samples could be handled separately by AD-EKI that generate dedicated inner ensembles for each outer sample at low additional cost. For the third group, which corresponds to outer samples near the main mass of the $y$ distribution, the improvement from grouping is less pronounced, because the global pooled posterior is already a reasonably good proposal. In principle, one could exploit this observation and use even cheaper proposals, such as directly reusing prior or outer samples for this group, thereby trading a small loss in accuracy for further savings in computational effort. In the present work, we do not pursue these variants and instead retain the grouped pooled posterior for all groups in order to maintain a consistent and simple framework.

\begin{figure}[H]
  \centering
  \begin{subfigure}[t]{0.32\textwidth}
    \centering
    \includegraphics[width=\linewidth]{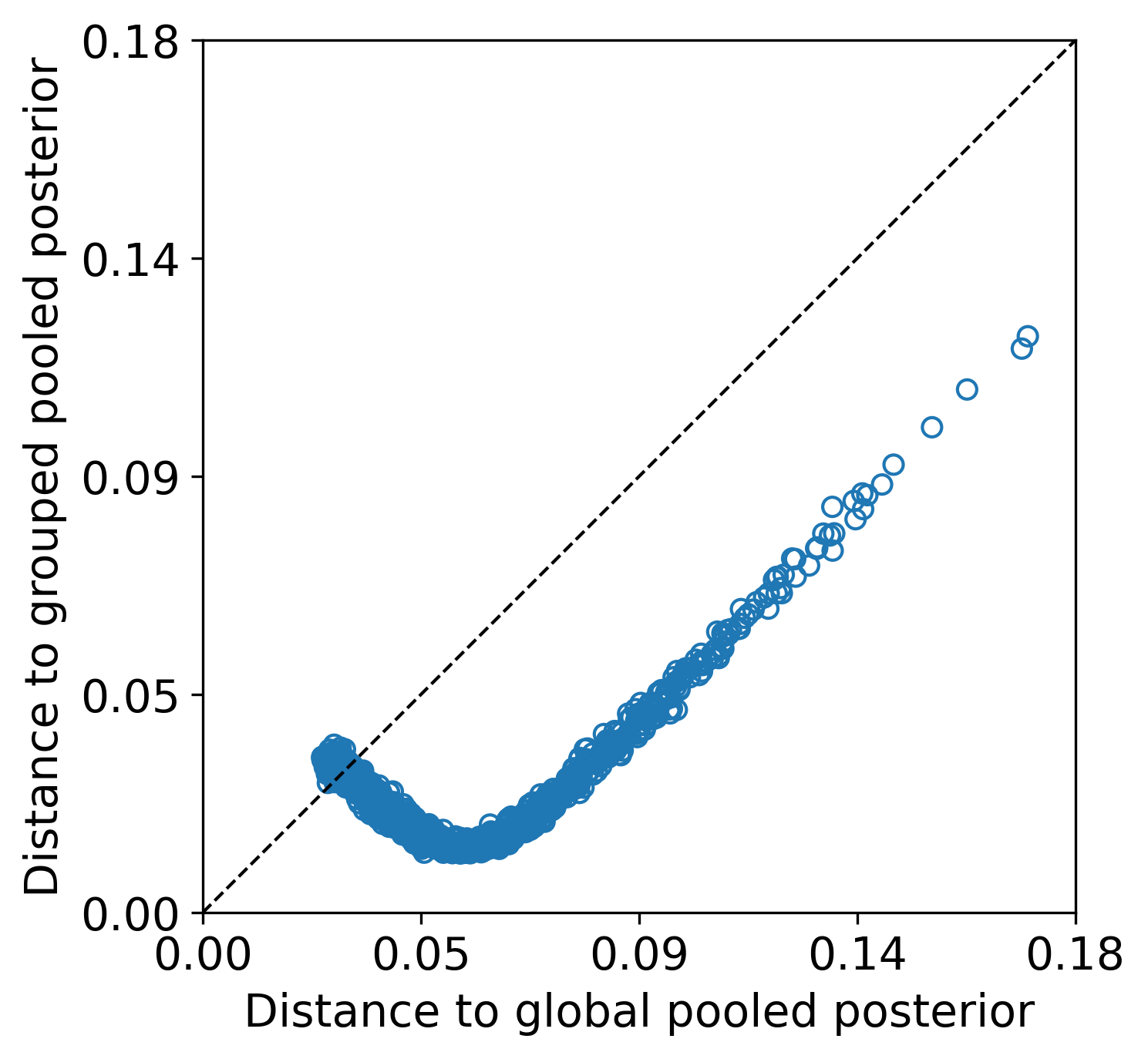}
    \caption{Group 1}
    \label{fig: group 1}
  \end{subfigure}
  \begin{subfigure}[t]{0.32\textwidth}
    \centering
    \includegraphics[width=\linewidth]{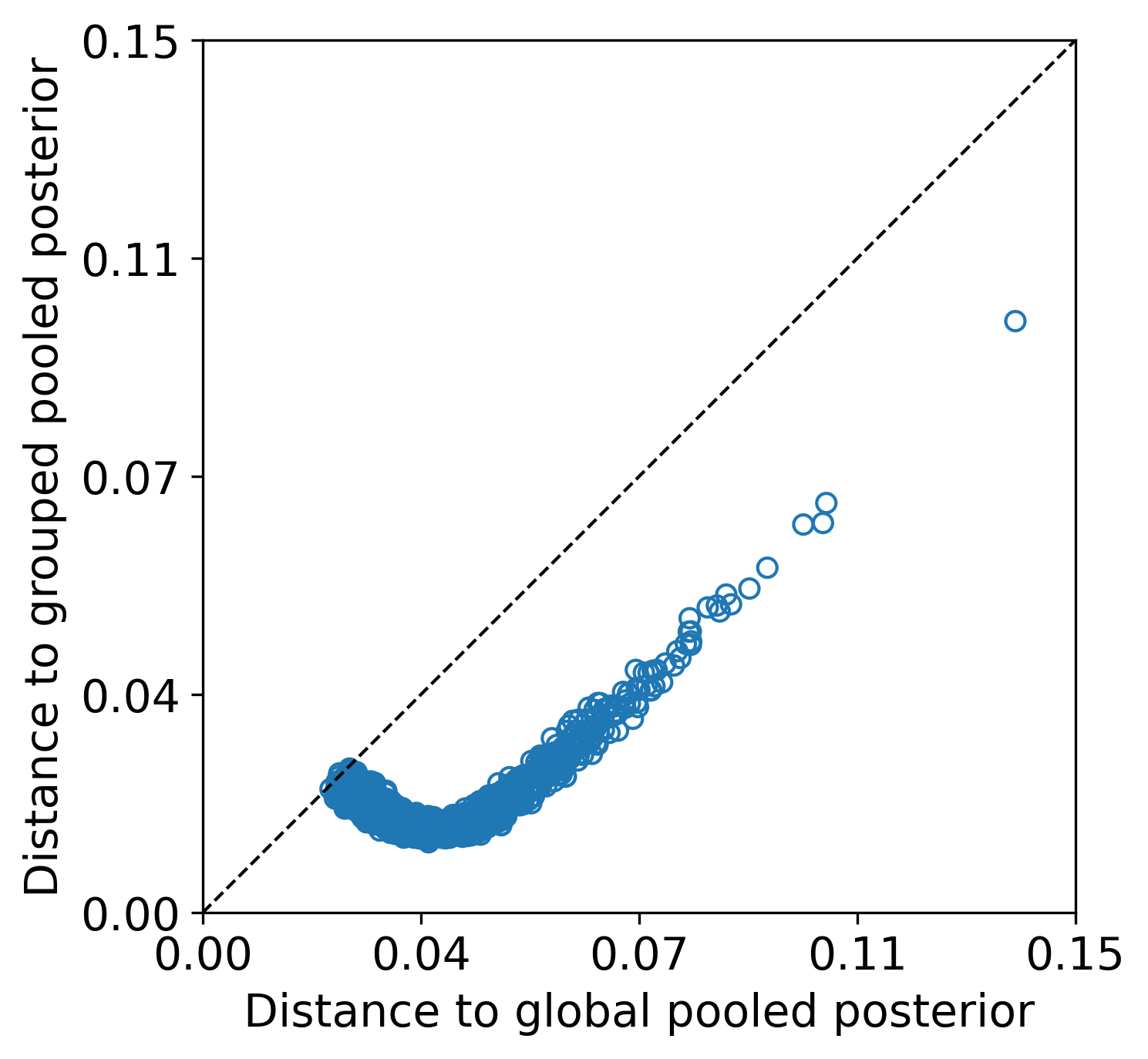}
    \caption{Group 2}
    \label{fig: group 2}
  \end{subfigure}
  \begin{subfigure}[t]{0.32\textwidth}
    \centering
    \includegraphics[width=\linewidth]{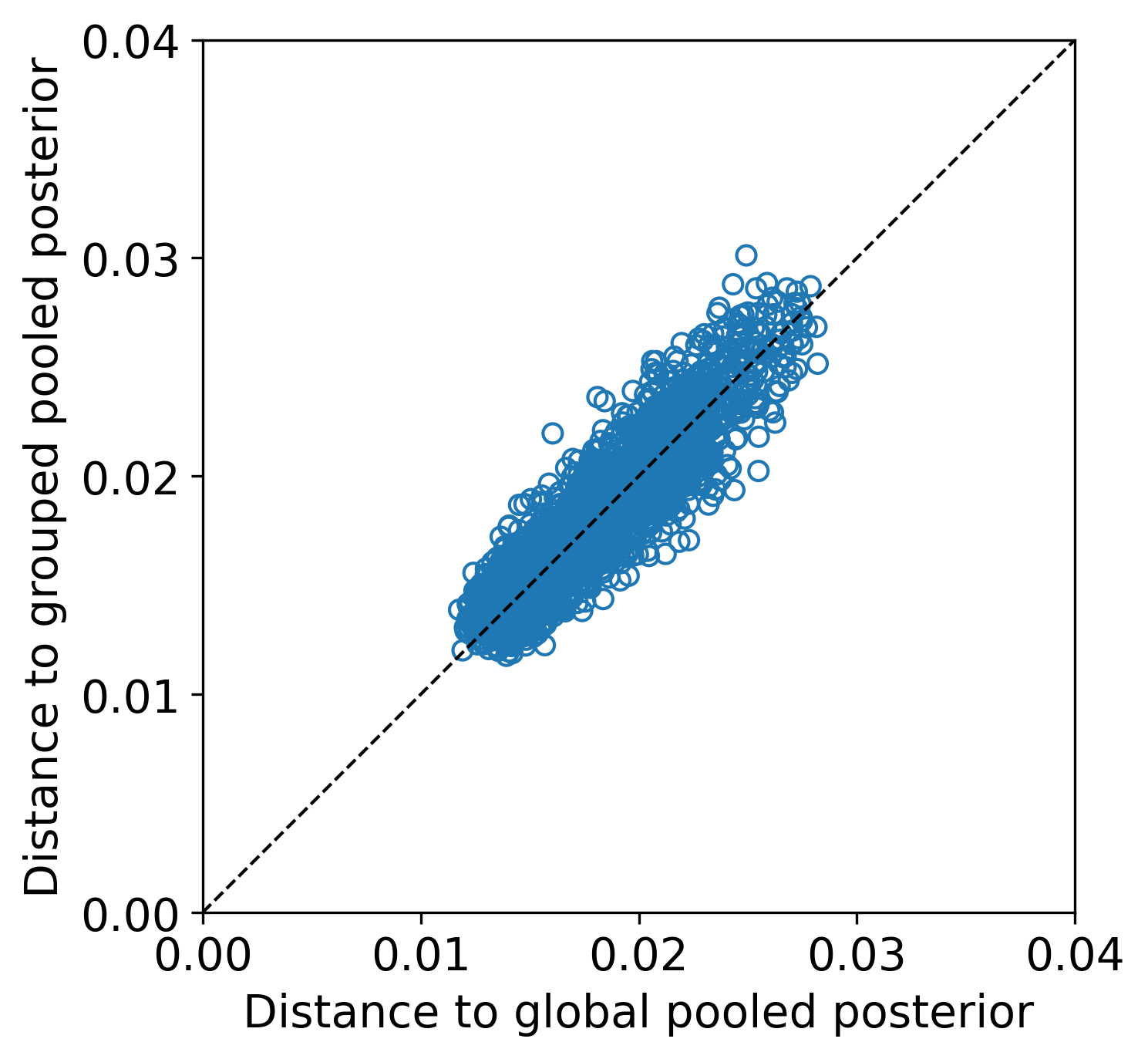}
    \caption{Group 3}
    \label{fig: group 3}
  \end{subfigure}
  \caption{Scatter plots comparing, for each outer sample, the distance from its individual posterior to the global geometric pooled posterior versus the distance to the grouped geometric pooled posterior.}
  \label{fig: no error reward map}
\end{figure}

We also provide a more detailed visual comparison for a single representative outer sample, see Fig.~\ref{fig: structural_corner}. We select five representative components of the high-dimensional neural-network parameter and plot a corner diagram of the approximate individual posterior, the global pooled posterior, and the grouped pooled posterior. All three distributions are approximated by EKI ensembles. In dimensions 2 and 3, the individual posterior and the global pooled posterior are strongly separated, with very little overlap, which would lead to poor importance sampling performance. The grouped pooled posterior substantially reduces this mismatch and is therefore expected to yield a much more reliable importance sampling estimate. In dimensions 0 and 14, the mismatch between the individual posterior and the global pooled posterior is more moderate but still visible, and again the grouped pooled posterior provides a noticeable improvement. In dimension 10, the global pooled posterior already matches the individual posterior quite well. Overall, this example confirms that the grouped pooled posterior systematically alleviates the mismatch between the global pooled posterior and the individual posteriors in the high-dimensional parameter space.

\begin{figure}[H]
    \centering
    \includegraphics[width=0.8\linewidth]{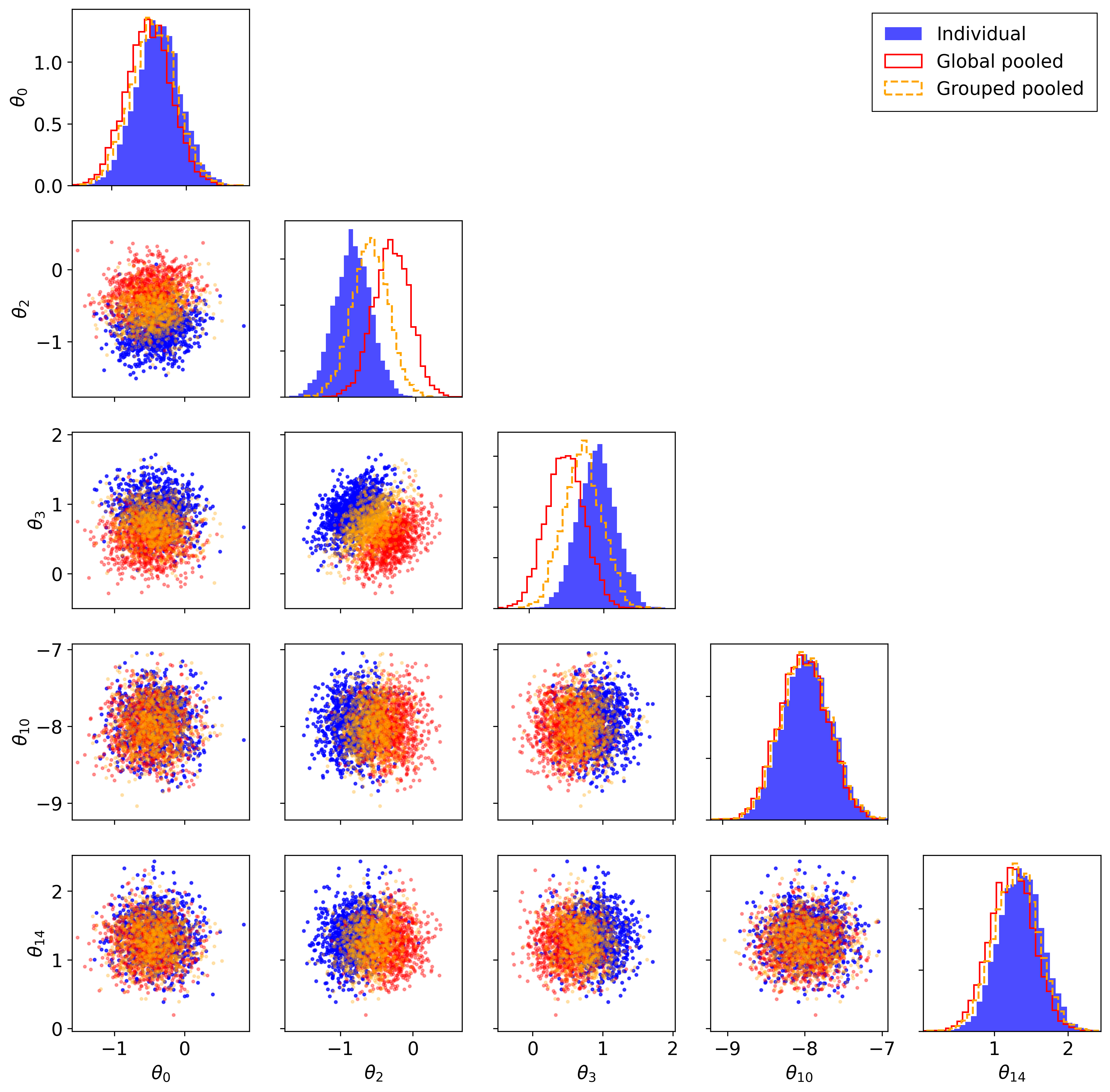}
    \caption{Corner plot of selected coordinates of the parameter vector, comparing individual, global pooled, and grouped pooled posterior samples.}
    \label{fig: structural_corner}
\end{figure}

Finally, Table~\ref{tab:Variance in importance sampling} summarizes the variability of the design-gradient estimator obtained by importance sampling. For this experiment, we fix the set of outer samples and generate multiple realizations of inner samples by changing only the random seed used in the inner sampling step, then compute the gradient of the EIG w.r.t. the design for each realization. The sample standard deviation of these gradient estimates is used as a measure of estimator quality. Without grouping, and using either 5000 or 15000 inner samples, the standard deviation of the gradient remains larger than 10, and increasing the number of inner samples produces almost no improvement. This is a typical symptom of weight degeneracy, where the target estimates no longer follow the standard Monte Carlo scaling rule with respect to the number of samples. Under the proposed grouping strategy, using the same total number of inner samples, 15000, the standard deviation of the gradient estimates is reduced to roughly one-third of the ungrouped value. This directly demonstrates that grouping the outer samples and using grouped pooled posteriors as proposals can substantially reduce the variance of importance sampling in the structural error setting.

\begin{table}[H]
  \centering
  \caption{Standard deviation in importance sampling}
  \label{tab:Variance in importance sampling}
  \begin{tabular}{cc}
    \toprule
      & Std \\ \midrule
    No group / 5000 sample size  & [12.6018, 10.1522] \\ 
    No group / 15000 sample size  & [12.7540, 10.5268] \\ 
    Group / 15000 sample size  & [4.0151, 3.4345] \\ \bottomrule
  \end{tabular}
\end{table}



\section{Conclusions}
\label{Conclusion}

In this work, we addressed the computational challenge of handling heterogeneous posterior families in nested BED by proposing a grouped geometric pooled posterior strategy. The key idea is to partition outer samples into groups so that each group admits a better-matched pooled proposal. Importantly, this grouped strategy retains the computational efficiency of shared-proposal approaches by avoiding additional forward-model evaluations. We systematically validated this framework through both parametric and structural model-discrepancy problems. In the Gaussian-linear parametric setting, we verified that the EKI sampler faithfully recovers the analytical pooled posterior, confirming its reliability. More importantly, in the high-dimensional structural error case involving neural-network corrections, our results revealed the limitations of a global "one-size-fits-all" proposal, which suffers from weight degeneracy due to poor coverage of tail distributions. By contrast, our grouping strategy improves the match between each group-specific pooled proposal and the corresponding posterior family within that group, and lowers the variance of the design-gradient estimator compared to the ungrouped baseline. Overall, the proposed method combines the robustness of outer-sample-specific proposals with the efficiency of shared-proposal approaches. These results suggest that GPP-EKI provides a scalable framework for BED in high-dimensional and computationally intensive settings where traditional nested sampling is intractable.


\appendix
\section{Appendix: proof of the equivalence}\label{apd: proof of proportion}

We provide the algebraic derivation for the mean-observation formulation of the geometrically pooled posterior. We show that the product of weighted Gaussian likelihoods is proportional to a single Gaussian likelihood acting on a precision-weighted mean.

Let $f(\boldsymbol\theta)$ denote the forward map. The geometrically pooled likelihood function is defined as:
\begin{equation}
\mathcal{L}(\boldsymbol\theta) \;\propto\; \prod_{i=1}^N p(\mathbf{y}_i | \boldsymbol\theta)^{\nu_i} \;\propto\; \exp\left( - \frac{1}{2} \sum_{i=1}^N \nu_i \, \| \mathbf{y}_i - f(\boldsymbol\theta) \|^2_{\boldsymbol\Sigma_i} \right),
\end{equation}
where we denote the inner term as $J(\boldsymbol\theta)$:
\begin{equation}
J(\boldsymbol\theta) = \sum_{i=1}^N \nu_i \, \| \mathbf{y}_i - f(\boldsymbol\theta) \|^2_{\boldsymbol\Sigma_i}.
\end{equation}
Expanding the quadratic term, we have:
\begin{equation}
\begin{aligned}
J(\boldsymbol\theta) &= \sum_{i=1}^N \nu_i \left( \mathbf{y}_i - f(\boldsymbol\theta) \right)^\top \boldsymbol\Sigma_i^{-1} \left( \mathbf{y}_i - f(\boldsymbol\theta) \right) \\
&= \sum_{i=1}^N \nu_i \left( \mathbf{y}_i^\top \boldsymbol\Sigma_i^{-1} \mathbf{y}_i - 2 f(\boldsymbol\theta)^\top \boldsymbol\Sigma_i^{-1} \mathbf{y}_i + f(\boldsymbol\theta)^\top \boldsymbol\Sigma_i^{-1} f(\boldsymbol\theta) \right).
\end{aligned}
\end{equation}
We now group the terms based on their dependence on $\boldsymbol\theta$:
\begin{equation}
\label{eq:expanded_J}
J(\boldsymbol\theta) = \underbrace{\sum_{i=1}^N \nu_i \mathbf{y}_i^\top \boldsymbol\Sigma_i^{-1} \mathbf{y}_i}_{\text{const.}}
- 2 f(\boldsymbol\theta)^\top \underbrace{\left( \sum_{i=1}^N \nu_i \boldsymbol\Sigma_i^{-1} \mathbf{y}_i \right)}_{\text{linear coef.}}
+ f(\boldsymbol\theta)^\top \underbrace{\left( \sum_{i=1}^N \nu_i \boldsymbol\Sigma_i^{-1} \right)}_{\text{quadratic coef.}} f(\boldsymbol\theta).
\end{equation}
The first term is constant with respect to $\boldsymbol\theta$ and can be absorbed into the normalization constant. We introduce the effective precision matrix $\boldsymbol\Sigma_\nu^{-1}$ and the precision-weighted mean $\bar{\mathbf{y}}_\nu$ as defined in the main text:
\begin{equation}
\boldsymbol\Sigma_\nu^{-1} = \sum_{i=1}^N \nu_i \boldsymbol\Sigma_i^{-1}, \qquad \boldsymbol\Sigma_\nu^{-1} \bar{\mathbf{y}}_\nu = \sum_{i=1}^N \nu_i \boldsymbol\Sigma_i^{-1} \mathbf{y}_i.
\end{equation}
Substituting these definitions into Eq.~\eqref{eq:expanded_J}:
\begin{equation}
J(\boldsymbol\theta) = \text{const.} - 2 f(\boldsymbol\theta)^\top \boldsymbol\Sigma_\nu^{-1} \bar{\mathbf{y}}_\nu + f(\boldsymbol\theta)^\top \boldsymbol\Sigma_\nu^{-1} f(\boldsymbol\theta).
\end{equation}
Now, consider the expansion of the squared Mahalanobis distance to the pooled mean $\bar{\mathbf{y}}_\nu$:
\begin{equation}
\begin{aligned}
\| \bar{\mathbf{y}}_\nu - f(\boldsymbol\theta) \|^2_{\boldsymbol\Sigma_\nu}
&= (\bar{\mathbf{y}}_\nu - f(\boldsymbol\theta))^\top \boldsymbol\Sigma_\nu^{-1} (\bar{\mathbf{y}}_\nu - f(\boldsymbol\theta)) \\
&= \bar{\mathbf{y}}_\nu^\top \boldsymbol\Sigma_\nu^{-1} \bar{\mathbf{y}}_\nu - 2 f(\boldsymbol\theta)^\top \boldsymbol\Sigma_\nu^{-1} \bar{\mathbf{y}}_\nu + f(\boldsymbol\theta)^\top \boldsymbol\Sigma_\nu^{-1} f(\boldsymbol\theta).
\end{aligned}
\end{equation}
Comparing the expanded forms, we observe that:
\begin{equation}
J(\boldsymbol\theta) = \| \bar{\mathbf{y}}_\nu - f(\boldsymbol\theta) \|^2_{\boldsymbol\Sigma_\nu} + C,
\end{equation}
where $C$ is a constant independent of $\boldsymbol\theta$. Therefore, the likelihood function satisfies:
\begin{equation}
\mathcal{L}(\boldsymbol\theta) \;\propto\; \exp\left( -\frac{1}{2} \| \bar{\mathbf{y}}_\nu - f(\boldsymbol\theta) \|^2_{\boldsymbol\Sigma_\nu} \right).
\end{equation}
This proves that the set of weighted observations is exactly equivalent to a single observation $\bar{\mathbf{y}}_\nu$ with noise covariance $\boldsymbol\Sigma_\nu$ for the purpose of parameter inference.

\section{Appendix: Derivation of the ESS formula}
\label{apd: ess}
We show that in the context of BED, with geometric pooled posterior as proposal of importance sampling, the effective sample size can be approximated as
\begin{equation}
  \mathrm{ESS}(\mathbf{y}_i)
  \approx
  J \exp\!\left(
    -\,(\mathbf{y}_i - \mathbf{y}_\nu)^\top R^{-1}\,\Sigma_{ff}\,R^{-1}(\mathbf{y}_i - \mathbf{y}_\nu)
  \right),
\end{equation}

Standard importance-sampling diagnostics relate the effective sample size (ESS) to the variability of the weights. A common starting point is the second-moment approximation \cite{sanz2020bayesian}:
\begin{equation}
\mathrm{ESS} \;\approx\; \frac{J}{\rho},
\label{eq:ess-rho}
\end{equation}
where $J$ is the number of importance samples and $\rho$ measures the $\chi^2$-divergence between between the target and proposal. Moreover, when both the target and proposal are Gaussian, $\rho$ admits an explicit closed form involving an exponential quadratic term, which is our starting point.

A further approximation is obtained under a log-normal model for the importance weights.
Specifically, \cite{neal2001annealed} shows that if the log-weights are (approximately) Gaussian, i.e., $\log\omega \approx \mathcal N(\mu,\sigma^2)$, then one may write
\begin{equation}
    \mathrm{ESS} \approx J \exp\left(-\mathrm{Var}(\log \omega)\right),
    \label{eq: ess step 1}
\end{equation}
where $\omega$ is the importance weights. In our setting, this equation is used as a general diagnostic: its accuracy depends on how well $\log\omega$ is approximated by a Gaussian random variable. We comment on this assumption at the end of this appendix.

In the context of importance sampling for BED, the target is the individual posterior $p(\boldsymbol\theta'|\mathbf{y}_i,\mathbf{d})$ associated with a certain outer sample $\mathbf{y}_i$, and the proposal is the geometric pooled posterior $p(\boldsymbol{\theta}'|Y,\mathbf{d})$. The importance weights $\omega(\boldsymbol\theta') \propto \frac{p(\boldsymbol\theta'|\mathbf{y}_i,\mathbf{d})}{p(\boldsymbol{\theta}'|Y,\mathbf{d})}$ are shown in \eqref{eq:snis} and its $\log$ value (with design $\mathbf{d}$ omitted for clarity):
\begin{equation}
    \begin{aligned}
        \log \omega(\boldsymbol\theta') &\;=\;\log  \frac{
        p(\mathbf{y}_i|\boldsymbol\theta' )
    }{
        \prod_{k=1}^{N} p(\mathbf{y}_k| \boldsymbol\theta')^{\nu_k}
    }+\mathrm{const}\\
    &\;=\; \log p(\mathbf{y}_i|\boldsymbol\theta' ) - \sum_{k=1}^N \nu_k \log p(\mathbf{y}_k| \boldsymbol\theta') +\mathrm{const}.
    \end{aligned}\label{eq: log weights}
\end{equation}

For a Gaussian observation model $p(\mathbf y_i | \boldsymbol\theta')=\mathcal N(\mathbf y_i; f(\boldsymbol\theta'),R)$, its $\log$ value is
\begin{equation}
    \begin{aligned}
        \log p(\mathbf y_i | \boldsymbol\theta') &= -\frac{1}{2} \left(\mathbf{y}_i - f(\boldsymbol{\theta}')\right)^\top R^{-1} \left(\mathbf{y}_i - f(\boldsymbol{\theta}')\right)\\
        &=\mathbf{y}_i^\top R^{-1} f(\boldsymbol{\theta}') - \frac{1}{2} f(\boldsymbol{\theta}')^\top R^{-1} f(\boldsymbol{\theta}') + \mathrm{const}(\mathbf{y_i}).
    \end{aligned}\label{eq: quandratic likelihood}
\end{equation}

Substituting Eq.~\eqref{eq: quandratic likelihood} into Eq.~\eqref{eq: log weights}:
\begin{equation}
    \begin{aligned}
        \log \omega(\boldsymbol\theta') &\;=\; \log p(\mathbf{y}_i|\boldsymbol\theta'_j ) - \sum_{k=1}^N \nu_k \log p(\mathbf{y}_k| \boldsymbol\theta'_j) +\mathrm{const}\\
        &\;=\; \left(\mathbf{y}_i-\sum_{k=1}^N \nu_k \mathbf{y}_k\right)^\top R^{-1} f(\boldsymbol{\theta}') - \frac{1}{2} \left(1-\sum_{k=1}^N \nu_k\right) f(\boldsymbol{\theta}') ^\top R^{-1} f(\boldsymbol{\theta}')  + \mathrm{const}
    \end{aligned}.\label{eq: log weight half}
\end{equation}

With
\[
1-\sum_{k=1}^N \nu_k = 0, \quad \sum_{k=1}^N \nu_k \mathbf{y}_k=\bar{\mathbf{y}}_\nu,
\]
Eq. \eqref{eq: log weight half} further simplifies to
\begin{equation}
    \log \omega(\boldsymbol\theta') \;=\; (\mathbf{y}_i - \bar{\mathbf{y}}_\nu)^\top R^{-1} f(\boldsymbol{\theta}')+\mathrm{const}.
\end{equation}

For the variance term $\mathrm{Var} \log (\omega)$ in Eq.~\eqref{eq: ess step 1}, the variance is taken with respect to $\boldsymbol\theta'$ from the proposal distribution, so the forward prediction $f(\boldsymbol{\theta}')$ is treated as a random vector. Notice $\log \omega(\boldsymbol\theta') $ is a linear function of $f(\boldsymbol{\theta}')$, thus its variance is written as
\begin{equation}
    \mathrm{Var} \log (\omega) = (\mathbf{y}_i - \bar{\mathbf{y}}_\nu)^\top R^{-1}\,\Sigma_{ff}\,R^{-1}(\mathbf{y}_i - \bar{\mathbf{y}}_\nu),
    \label{eq: var log weight}
\end{equation}
with $\Sigma_{ff}$ being the covariance matrix of $f(\boldsymbol{\theta}')$:
\[
\Sigma_{ff}=\frac{1}{J-1} \sum_{j=1}^J \left(f(\boldsymbol{\theta}_j')-\bar{f}\right)\left(f(\boldsymbol{\theta}_j')-\bar{f}\right)^\top, \quad \bar{f}=\frac{1}{J}\sum_{j=1}^J f(\boldsymbol{\theta}_j').
\]

Substituting Eq.~\eqref{eq: var log weight} into Eq.~\eqref{eq: ess step 1}, then the desired formula is obtained.

Here we discuss the exactness. It is important to distinguish the algebraic derivation from the statistical approximation in our result.

First, the derivation of the log-weight variance is exact. For a Gaussian observation model and normalized geometric pooling weights ($\sum_k \nu_k = 1$), the log-weight admits the exact representation
$\log\omega(\theta') = a^\top f(\theta') + \mathrm{const}$ with $a = R^{-1}(\mathbf y_i-\bar{\mathbf{y}}_\nu)$, and hence
$\mathrm{Var}(\log\omega)=a^\top \Sigma_{ff} a$ holds as an identity.

The only approximation in \eqref{eq: ess step 1} is the log-normal (Gaussian log-weight) assumption. This assumption is exact in linear--Gaussian settings, e.g., when $f(\boldsymbol\theta')$ is linear and the proposal distribution for $\theta'$ is Gaussian, in which case $a^\top f(\boldsymbol\theta')$ (and thus $\log\omega$) is Gaussian. It remains a useful heuristic as long as the $a^\top f(\boldsymbol\theta')$ is approximately Gaussian under the proposal, even if the full distribution of $f(\boldsymbol\theta')$ is non-Gaussian.

\section{Appendix: Proof of Covariance contraction}
\label{apd: lower covariance}

\begin{proposition}
Let $\theta \sim \mathcal{N}(m_0, P_0)$ and $Y = A\theta + \varepsilon$ with $\varepsilon \sim \mathcal{N}(0, \Sigma)$ and $\Sigma \succ 0$. Denote $f := A\theta$. 
Let $P_{ff} := \mathrm{Cov}(f) = A P_0 A^\top$ and $\Sigma_{ff} := \mathrm{Cov}(f \mid Y) = A P_{\mathrm{new}} A^\top$, where $P_{\mathrm{new}}$ is the posterior covariance of $\theta$ given $Y$. Then
\begin{equation}
\Sigma_{ff} \preceq P_{ff}.
\end{equation}
\end{proposition}

\begin{proof}
For the linear--Gaussian model, the posterior covariance is
\[
P_{\mathrm{new}} = P_0 - P_0 A^\top (A P_0 A^\top + \Sigma)^{-1} A P_0.
\]
Hence,
\[
P_0 - P_{\mathrm{new}} = P_0 A^\top (A P_0 A^\top + \Sigma)^{-1} A P_0 \succeq 0,
\]
so $P_{\mathrm{new}} \preceq P_0$. Multiplying by $A$ and $A^\top$ preserves the Loewner order, yielding
\[
\Sigma_{ff} = A P_{\mathrm{new}} A^\top \preceq A P_0 A^\top = P_{ff}.
\]
\end{proof}


\bibliographystyle{siamplain}
\bibliography{references}
\end{document}